\documentclass[11pt]{article}
\usepackage[top=1in,left=1in,right=1in,bottom=1in]{geometry}

\usepackage{style}
\usepackage{times}
\usepackage{balance}  

\newcommand{\wcost}{\omega}

\newcommand{\argmax}{\operatorname{arg\,max}}

\newcommand{\depth}{span}

\newcommand{\id}[1]{\ifmmode\mathit{#1}\else\textit{#1}\fi}
\newcommand{\const}[1]{\ifmmode\mbox{\textc{#1}}\else\textsc{#1}\fi}
\newcommand{\anp}{Asymmetric NP}

\newcommand{\ourcomp}{$k$-d grid}
\newcommand{\kdcomp}[1]{$#1$d grid}
\newcommand{\sqr}{square grid}
\newcommand{\ourcompfull}{$k$-d grid computation structure}

\newcommand{\sqrfull}{square grid computation structure}

\begin{document}
\title{Improved Parallel Cache-Oblivious Algorithms for \\Dynamic Programming and Linear Algebra}
\author{Guy Blelloch\\Carnegie Mellon University \and Yan Gu\\Massachusetts Institute of Technology}
\date{}

\maketitle

\begin{abstract}
Emerging non-volatile main memory (NVRAM) technologies provide byte-addressability, low idle power, and improved memory-density, and are likely to be a key component in the future memory hierarchy.
However, a critical challenge in achieving high performance is in accounting for the asymmetry that NVRAM writes can be significantly more expensive than NVRAM reads.

In this paper, we consider a large class of cache-oblivious algorithms for dynamic programming (DP) and linear algebra, and try to reduce the writes in the asymmetric setting while maintaining high parallelism.
To achieve that, our key approach is to show the correspondence between these problems and an abstraction for their computation, which is referred to as the $k$-d grids.
Then by showing lower bound and new algorithms for computing $k$-d grids, we show a list of improved cache-oblivious algorithms of many DP recurrences and in linear algebra in the asymmetric setting, both sequentially and in parallel.

Surprisingly, even without considering the read-write asymmetry (i.e., setting the write cost to be the same as the read cost in the algorithms), the new algorithms improve the existing cache complexity of many problems.
We believe the reason is that the extra level of abstraction of $k$-d grids helps us to better
understand the complexity and difficulties of these problems.
We believe that the novelty of our framework is of interests and leads to many new questions for future work.
\end{abstract} 

\section{Introduction}

The \emph{ideal-cache model}~\cite{Frigo99} is widely used in designing algorithms that optimize the communication between CPU and memory.
The model is comprised of an unbounded memory and a cache of size $M$.
Data are transferred between the two levels using cache lines of size $B$, and all computation occurs on data in the cache.
An algorithm is \emph{cache-oblivious} if it is unaware of both $M$ and $B$.
The goal of designing such algorithms is to reduce the \emph{cache complexity}\footnote{In this paper, we refer to it as \emph{symmetric cache complexity} to distinguish from the case when reads and writes have different costs.} (or the \emph{I/O cost} indistinguishably) of an algorithm, which is the number of cache lines transferred between the cache and the main memory assuming an optimal (offline) cache replacement policy.
Sequential cache-oblivious algorithms are flexible and portable, and adapt to all levels of a multi-level memory hierarchy.
Such algorithms are well studied~\cite{arge2004cache,brodal2004cache,demaine2002cache}, and in many cases they asymptotically match the best cache complexity for cache-aware algorithms.
Regarding parallelism, Blelloch et al.~\cite{blelloch2010low} suggest that analyzing the depth and sequential cache complexity of an algorithm is sufficient for deriving upper bounds on parallel cache complexity.

\begin{table*}[t]
  \centering
\def\arraystretch{1.4}
    \begin{tabular}{clcc}
    \toprule
    \multicolumn{1}{c}{\bf \multirow{2}[0]{*}{Dimension}} & \multicolumn{1}{c}{\bf \multirow{2}[0]{*}{Problems}} & \multicolumn{2}{c}{\bf Cache Complexity} \\
          &       & \multicolumn{1}{l}{\bf Symmetric} & \multicolumn{1}{l@{}}{\bf Asymmetric} \\
    \midrule
    $k=2$   & LWS/GAP*/RNA/knapsack recurrences      & $\displaystyle \Theta\left(\frac{C}{B{M}}\right)$      & $\displaystyle \Theta\left(\frac{\wcost{}^{1/2} C}{B{M}}\right)$ \medskip\\
    \multirow{2}[0]{*}{$k=3$}   & Combinatorial matrix multiplication,      \vspace{-.5em}  & \multirow{2}[0]{*}{$\displaystyle \Theta\left(\frac{C}{B\sqrt{M}}\right)$}      & \multirow{2}[0]{*}{$\displaystyle \Theta\left(\frac{\wcost{}^{1/3} C}{B\sqrt{M}}\right)$} \\
    &Kleene's algorithm (APSP), Parenthesis recurrence&&\vspace{.5em}\\
    \bottomrule
    \end{tabular}%
\smallskip
\caption[I/O costs of cache-oblivious algorithms]{\label{tab:overview} Cache complexity of the algorithms based on the \ourcompfull{s}.  Here $C$ is the number of algorithmic instructions in the corresponding computation.  (*) For the GAP recurrence, the upper bounds have addition terms as shown in Section~\ref{sec:gap}.}
\vspace{-1em}
\end{table*}%

Recently, emerging non-volatile main memory (NVRAM) technologies, such as Intel's Optane DC Persistent Memory, are readily available on the market, and provide byte-addressability, low idle power, and improved memory-density.
Due to these advantages, NVRAMs are likely to be the dominant main memories in the near future, or at least be a key component in the memory hierarchy.
However, a significant programming challenge arises due to an underlying asymmetry between reads and writes---reads are much cheaper than writes in terms of both latency and throughput.
This property requires researchers to rethink the design of algorithms and software, and optimize the existing ones accordingly to reduce the writes.
Such algorithms are referred to as the \defn{write-efficient algorithms}~\cite{GuThesis}.

Many cache-oblivious algorithms are affected by this challenge.
Taking matrix multiplication as an example, the cache-aware tiling-based algorithm~\cite{AggarwalV88} uses $\Theta(n^3/B\sqrt{M})$ cache-line reads and $\Theta(n^2/B)$ cache-line writes for square matrices with size $n$-by-$n$.
The cache-oblivious algorithm~\cite{Frigo99}, despite the advantages described above, uses $\Theta(n^3/B\sqrt{M})$ cache-line reads and writes.
When considering the more expensive writes, the cache-oblivious algorithm is no longer asymptotically optimal.
Can we asymptotically improve the cache complexity of these cache-oblivious algorithms?
Can they match the best counterpart without considering cache-obliviousness?
These remain to be \emph{open} problems in the very beginning of the study of write-efficiency of algorithms~\cite{BFGGS15,carson2016write}.

In this paper, we provide the answers to these questions for a large class of cache-oblivious algorithms that have computation structures similar to matrix multiplication and can be coded up in nested for-loops.
Their implementations are based on a divide-and-conquer approach that partitions the ranges of the loops and recurses on the subproblems until the base case is reached.
Such algorithms are in the scope of dynamic programming (e.g., the LWS/GAP/RNA/Parenthesis problems) and linear algebra (e.g., matrix multiplication, Gaussian elimination, LU decomposition)~\cite{Frigo99,chowdhury2006cache,chowdhury2010cache,chowdhury2016autogen,blelloch2010low,itzhaky2016deriving,tithi2015high,tang2017,womble1993beyond,toledo1997locality}.

Since we try to cover many problems and algorithms, in this paper we propose a level of abstraction of the computation in these cache-oblivious algorithms, which is referred to as the \ourcompfull{s} (short for the \ourcomp{s}).
A more formal definition is given Section~\ref{sec:kdcomp}.
This structure and similar ones are first used by Hong and Kung~\cite{HK81} (implicitly) in their seminal paper in 1981, and then by a subsequence of later work (e.g.,~\cite{BallardDHS10,AggarwalCS90,IronyTT04,BallardDHS11,chowdhury2010cache}), mostly on analyzing the lower bounds of matrix multiplication and linear algebra problems in a variety of settings.
In this paper, we show the relationship of the \ourcomp{s} and many other dynamic programming problems, and new results (algorithms and lower bounds) related to the \ourcomp{s}.

The first intellectual contribution of this paper is to draw the connection between many dynamic programming (DP) problems and the \ourcomp{s}.
Previous DP algorithms are usually designed and analyzed based on the number of nested loops, or the number of the dimensions of which the input and output are stored and organized.
However, we observe that the key underlying factor in determining the cache complexity of these computations is the \textbf{number of input entries} involved in each basic computation cell, and such relationship will be defined formally later in Section~\ref{sec:kdcomp}.
A few examples (e.g., matrix multiplication, tensor multiplication, RNA and GAP recurrences) are also provided in Section~\ref{sec:kdcomp} to illustrate the idea.
This property is reflected by the nature of the \ourcomp{s}, and the correspondence between the problems and the \ourcomp{s} is introduced in Section~\ref{sec:numerical} and~\ref{sec:dp}.
We note that such relationship can be much more complicated than the linear algebra algorithms, and in many cases the computation of one algorithm consists of many (e.g., $O(n)$) \ourcomp{s} associated with restrictions of the order of computing.

The second intellectual contribution of this paper is a list of new results for the \ourcomp{s}.
We first discuss the lower bounds to compute such \ourcomp{s} considering the asymmetric cost between writes and reads (in Section~\ref{sec:lower}).
Based on the analysis of the lower bounds, we also propose algorithms with the matching bound to compute a \ourcomp{} (in Section~\ref{sec:algorithm}).
Finally, we also show how to parallelize the algorithm in Section~\ref{sec:para}.
We note that the approach for parallelism is independent of the asymmetric read-write cost, so the parallel algorithms can be applied to both symmetric and asymmetric algorithms.

In summary, we have shown the correspondence between the problems and the \ourcomp{s}, new lower and upper cache complexity bounds for computing the \ourcomp{s} in asymmetric setting, and parallel algorithms in both symmetric and asymmetric settings.
Putting all pieces together, we can show lower and upper cache complexity bounds of the original problems in both the symmetric and asymmetric settings, as well as spans (length of dependence) for the algorithms.
The cache complexity bounds are summarized in Table~\ref{tab:overview}, and the results for the asymmetric setting answer the open problem in~\cite{BFGGS15}.
The span bound is analyzed for each specific problem and given in Section~\ref{sec:dp},~\ref{sec:numerical} and appendices.

Surprisingly, even without considering the read-write asymmetry (i.e., setting the write cost to be the same as the read cost in the algorithms), the new algorithms proposed in this paper improve the existing cache complexity of many DP problems.
We believe the reason is that the extra level of abstraction of \ourcomp{s} helps us to better understand the complexity and difficulties of these problems.
Since \ourcomp{s} are used as a tool for lower bounds, they decouple the computation structures from the complicated data dependencies, which exposes some techniques to improve the bounds that were previously obscured.
Also, \ourcomp{s} reveal the similarities and differences between these problems, which allows the optimizations in some algorithms to apply to other problems.

\medskip
In summary, we believe that the framework for analyzing cache-oblivious algorithms based on \ourcomp{s} provides a better understanding of these algorithms.  In particular, the new theoretical results in this paper include:
\begin{itemize}
  \item We provide write-efficient cache-oblivious algorithms (i.e., in the asymmetric setting) for all problems we discussed in this paper, including matrix multiplication and several linear algebra algorithms, all-pair shortest-paths, and a number of dynamic programming recurrences.  If a write costs $\wcost{}$ times more than a read (the formal computational model shown in Section~\ref{sec:prelim}), the asymmetric cache complexity is improved by a factor of $\Theta(\wcost{}^{1/2})$ or $\Theta(\wcost{}^{2/3})$ on each problem compared to the best previous results~\cite{BFGGS15arxiv}.  In some cases, we show that this improvement is optimal under certain assumptions (the CBCO paradigm, defined in Section~\ref{sec:asym-lower-bound}).
  \item We show algorithms with improved symmetric cache complexity on many problems, including the GAP recurrence, protein accordion folding, and the RNA recurrence.  We show that the previous cache complexity bound $O(n^3/B\sqrt{M})$ for the GAP recurrence and protein accordion folding is not optimal, and we improve the bound to $O(n^2/B\cdot(n/M+\log \min\{n/\sqrt{M},\sqrt{M}\}))$ and $\Theta(n^2/B\cdot(1+n/M))$ respectively\footnote{The improvement is $O(\sqrt{M})$ from an asymptotic perspective (i.e., $n$ approaching infinity).  For smaller range of $n$ that $O(\sqrt{M})\le n\le O(M)$, the improvement is $O(n/\sqrt{M}/\log(n/\sqrt{M}))$ and $O(n/\sqrt{M})$ respectively for the two cases. (The computation fully fit into the cache when $n<O(\sqrt{M})$.)}.  For RNA recurrence, we show an optimal cache complexity of $\Theta(n^4/BM)$, which improves the best existing result by $\Theta(M^{3/4})$.
  \item We show the first race-free linear-\depth{} cache-oblivious algorithms solving all-pair shortest-paths, LWS recurrences, and protein accordion folding.  Some previous algorithms~\cite{tang2015cache,dinh2016extending} have linear span, but they are not race-free and rely on a much stronger model (discussion in Section~\ref{sec:prelim}).  Our approaches are under the standard nested-parallel model, race-free, and arguably simpler.  Our algorithms are not in-place, but we  discuss in Section~\ref{sec:para-sym} about the extra storage needed.
\end{itemize}

We believe that the analysis framework is concise.  In this single paper, we discuss the lower bounds and parallel algorithms on a dozen or so computations and DP recurrences, which can be further applied to dozens of real-world problems\footnote{Like in this paper we abstract the ``2-knapsack recurrence'', which fits into our \ourcompfull{} and applies to many algorithms.}.  The results are shown in both settings with or without considering the asymmetric cost between reads and writes.

\section{Preliminaries and Related Work}\label{sec:prelim}

\myparagraph{Ideal-cache model and cache-oblivious algorithms}
In modern computer architecture, a memory access is much more expensive compared to an arithmetic operation due to larger latency and limited bandwidth (especially in the parallel setting).
To capture the cost of an algorithm on memory access, the \emph{ideal-cache model}, a widely-used cost model, is a {two-level memory model} comprised of
an unbounded memory and a cache of size $M$.\footnote{In this paper, we often assume the cache size to be $O(M)$ since it simplifies the description and only affects the bounds by a constant factor.}
Data are transferred between the two levels using cache lines of size $B$, and all computation occurs on data in the cache.
The {cache complexity} (or the {I/O cost} indistinguishably) of an algorithm is the number of cache lines transferred between the cache and the main memory assuming an optimal (offline) cache replacement policy.
An algorithm on this model is \defn{cache-oblivious} with the additional feature that it is not aware of the value of $M$ and $B$.
In this paper, we refer to this cost as the \defn{symmetric cache complexity} (as opposed to asymmetric memory as discussed later).
Throughout the paper, we assume that the input and output do not fit into the cache since otherwise the problems become trivial.
We usually make the tall-cache assumption that $M=\Omega(B^2)$, which holds for real-world hardware and is used in the analysis in Section~\ref{sec:rna}.

\myparagraph{The nested-parallel model and work-\depth{} analysis}
In this paper, the parallel algorithms are within the standard nested-parallel model, which is a computation model and provides an easy analysis of the work-efficiency and parallelism.
In this model, a computation starts and ends with a single root task.
Each task has a constant number of registers and runs a standard instruction set from a random access machine, except it has one additional instruction called FORK, which can create two independent tasks one at a time that can be run in parallel.
When the two tasks finish, they join back and the computation continues.

A computation can be viewed as a (series-parallel) DAG in the standard way.
The cost measures on this model are the work and \depth{}---\emph{work} $W$ to be the total number of operations in this DAG and span (depth) $D$ equals to the longest path in the DAG.
The randomized work-stealing scheduler can execute such a computation on the PRAM model with $p$ processors in $W/p+O(D)$ time with high probability~\cite{blumofe1999scheduling}.
All algorithms in this paper are \emph{race-free}~\cite{feng1999efficient}---no logically parallel parts of an algorithm access the same memory location and one of the accesses is a write. 
Here we do not distinguish the extra write cost for asymmetric memory on $W$ and $D$ to simplify the description of the results, and we only capture this asymmetry using cache complexity.


Regarding parallel cache complexity, Blelloch et al.~\cite{blelloch2010low} suggest that analyzing the \depth{} and sequential cache complexity of an algorithm is sufficient for deriving upper bounds on parallel cache complexity.
In particular, let $Q_1$ be the sequential cache complexity.
Then for a $p$-processor shared-memory machine with private caches (i.e., each processor has its own cache) using a work-stealing scheduler, the total number of cache misses $Q_p$ across all processors is at most $Q_1 + O(pDM/B)$ with high
probability~\cite{Acar02}.
For a $p$-processor shared-memory machine with a shared cache of size $M + pBD$ using a parallel-depth-first (PDF)
scheduler, $Q_p\le Q_1$~\cite{BlGi04}.
We can extend these bounds to multi-level hierarchies of private or shared caches, respectively~\cite{blelloch2010low}.

\myparagraph{Parallel and cache-oblivious algorithms for dynamic programming and linear algebra}
Dynamic Programming (DP) is an optimization strategy that decomposes a problem into subproblems with optimal substructure.
It has been studied for over sixty years~\cite{bellman1957dynamic,aho1974design,CLRS}.
For the problems that we consider in this paper,
the parallel DP algorithms were already discussed by a rich literature in the eighties and nighties (e.g., \cite{galil1989speeding,galil1994parallel,eppstein1989parallel,huang1992sublinear,huang1994parallel,rytter1988efficient}).
Later work not only considers parallelism, but also optimizes symmetric cache complexity~(e.g., \cite{Frigo99,chowdhury2006cache,chowdhury2010cache,chowdhury2016autogen,blelloch2010low,itzhaky2016deriving,tithi2015high,tang2017,tang2015cache,dinh2016extending,solomonik2013minimizing,Chowdhurythesis}).
The algorithms in linear algebra that share the similar computation structures (but with different orders in the computation) are also discussed~(e.g., \cite{chowdhury2010cache,dinh2016extending,womble1993beyond,toledo1997locality,blumofe96,demmel2013communication,ballard2014communication,koanantakool2016communication}).

\myparagraph{Problem definitions}
Since we are showing many optimal cache-oblivious algorithms, we need formal problem definitions.
It is hard to show general lower bounds that any type of operations is allowed.
For example, for matrix multiplication on a semiring,
the only known lower bound of operations is just $\Omega(n^{3-o(1)})$ for Boolean matrix multiplication assuming SETH (more details of fine-grain complexity in~\cite{Williams2018some}).
Here we make no assumptions of the set of the ring other than ``$+$'' and ``$\times$''  to be atomic using unit cost and unable to be decomposed or batched (i.e., using integer tricks).
We borrow the term \emph{combinatorial matrix multiplication} to indicate this specific problem.
Such an algorithm requires $\Theta(n^3)$ operations on square matrices of size $n$ ($n^2$ inner products).
Regarding dynamic programming in Section~\ref{sec:dp}, we discuss the recurrences rather than the problems, and make assumptions shown in Section~\ref{sec:kdcomp} and Section~\ref{sec:dp}.

\myparagraph{Algorithms with asymmetric read and write costs}
Intel has already announced the new
product of the Optane DC Persistent Memory, which can be
bought from many retailers. The new memories sit on the
main memory bus and are byte-addressable. As opposed to
DRAMs, the new memories are persistent, so we refer to
them as non-volatile RAMs (NVRAMs). In addition, compared
to DRAMs, NVRAMs require significantly lower energy,
and have good read latencies and higher density. Due
to these advantages, NVRAMs are likely to be the dominant
main memories in the near future, or at least be a key component
in the memory hierarchy. However, a new property
of NVRAMs is the asymmetric read and write cost---write operations
are more expensive than reads regarding energy,
bandwidth, and latency (benchmarking results in~\cite{van2019persistent}). This property requires researchers
to rethink the design of algorithms and software, and motivates the need for \emph{write-efficient algorithms}~\cite{GuThesis}
that reduce the number of writes compared to existing algorithms.

Blelloch et al.~\cite{BBFGGMS16,BFGGS15,blelloch2016efficient}
formally defined and analyzed several sequential and parallel computation models that take
asymmetric read-write costs into account.
The model Asymmetric RAM (ARAM) extends the two-level memory model and contains a parameter $\wcost$, which corresponds to the cost of a write relative to a read to the non-volatile main memory.
In this paper, we refer to the \defn{asymmetric cache complexity} $Q$ as the number of write transfers to the main memory multiplied by $\wcost$, plus the number of read transfers.
This model captures different system consideration (latency, bandwidth, or energy) by simply plugging in a different value of $\wcost$, and also allows algorithms to be analyzed theoretically and practically.
Similar scheduling results (upper bounds) on parallel running time and cache complexity are discussed in~\cite{BBFGGMS16,blelloch2016efficient} based on work $W$, \depth{} $D$ and asymmetric cache complexity $Q$ of an algorithm.
Based on this idea, many interesting algorithms and lower bounds are designed and analyzed by various recent works~\cite{BBFGGMS16,BFGGS15,blelloch2016efficient,blelloch2016parallel,Jacob17,BBFGGMS18implicit,blelloch2018parallelfull,blelloch2018PPMM,Gu2018}.

\hide{
We assume that the input and output size is larger than $\wcost{}M$ in the asymmetric setting (e.g.,  the input is not smaller or close to the cache size), since it is more theoretically and practically interesting.
Otherwise the asymmetric cache complexity for all algorithms in this paper is $\Theta(\wcost{}n^{k-1}/B)$ trivially, the cost for output.
}
In the analysis, we always assume that the input size is much larger than the cache size (which is usually the case in practice).
Otherwise, both the upper and the lower bounds on cache complexity also include the term for output---$\wcost$ times the output size.
For simplicity, this term is ignored in the asymptotic analysis. 

Carson et al.~\cite{carson2016write} also discussed algorithms using less writes.
Their results are under some different assumptions that disallow the use of more reads, and we discuss how the assumptions affect the algorithms in the Appendix~\ref{sec:carson}.

\begin{figure*}[tp]
\centering
  \includegraphics[width=.6\linewidth]{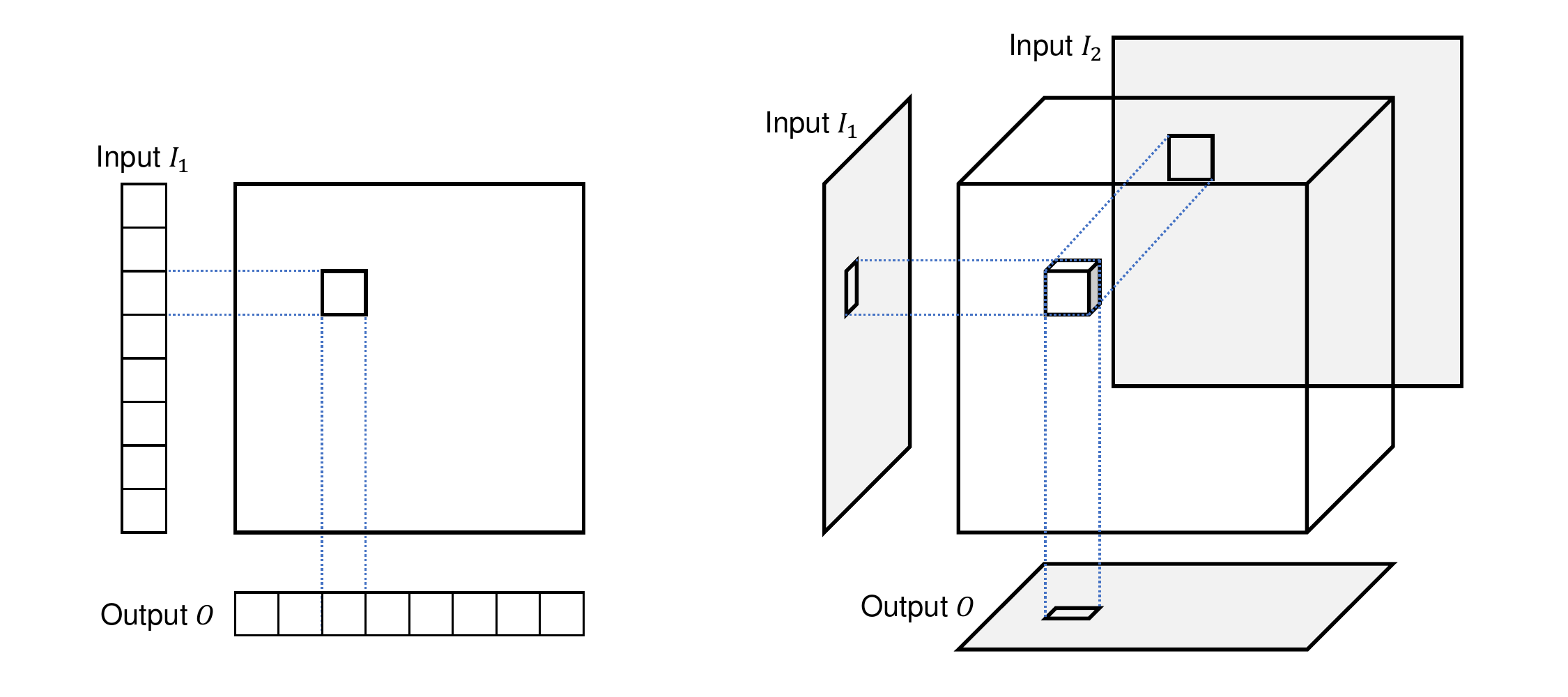}
  \vspace{-1em}
  \caption{An illustration of a $2$d and a \kdcomp{3}.  The left figure shows the $2$d case where the input $I_1$ and output $O$ are $1$d arrays, and each computation cell $g(\cdot)$ requires exactly one entry in $I_1$ as input, and update one entry in $O$.  For the $3$d case on the right, the inputs and output are $2$d arrays, and each computation cell $g(\cdot)$ requires one entry from input $I_1$ and one from input $I_2$. The input/output entries of each cell are the projections of this cell on different $2$d arrays. }\label{fig:grid}
\end{figure*}

\myparagraph{Discussions of previous work}
We now discuss several possible confusions of this paper.

The \ourcompfull{} in this paper is similar to the structure in Hong and Kung~\cite{HK81}, and some subsequence work in linear algebra (e.g.,~\cite{BallardDHS10,AggarwalCS90,IronyTT04,BallardDHS11,chowdhury2010cache}).
Several recent papers on DP algorithms are also based on grid structure (e.g., \cite{chowdhury2016autogen,itzhaky2016deriving,tithi2015high}), but the definitions in those paper are different from the \ourcomp{s} in this paper.
In the \ourcomp{}, the dimension of the grid is related to the number of entries per basic computation unit (formal definition in Section~\ref{sec:kdcomp}), not the dimension of the input/output arrays.
However, in the special case when the number of input entries per basic computation cell is the same as the dimension of the input/output arrays, the analysis based on \ourcomp{} provides the \textbf{\emph{same}} sequential symmetric cache complexity as the previous work.
One such example is matrix multiplication~\cite{Frigo99},
\hide{
 \guy{I'm not sure the rest of this sentence is necessary}, and we recommend the reader to use it as an instantiation to understand our improved parallel and asymmetric algorithms in Section~\ref{sec:algorithm} and~\ref{sec:para}.
\guy{the rest of this paragraph does not make much sense to me.}
This is also the case for other linear-algebra algorithms in the full version of this paper, and the LWS and Parenthesis recurrences in Section~\ref{sec:lws} and~\ref{sec:parenthesis}.}
and in these cases we still provide new lower and upper asymmetric cache bounds as well as parallel approaches (new span bounds).
For other problems (GAP, RNA, protein accordion folding, knapsack), the bounds in the symmetric setting are also improved.


\hide{
Since we believe that the abstraction and definition are new, the proof of lower bounds in Section~\ref{sec:lower} focuses on the simplest sequential setting with a limited-size cache, which is different from the parallel or distributed settings with infinite-size local memory discussed in previous work (e.g., \cite{BallardDHS10,AggarwalCS90,IronyTT04,BallardDHS11}).
It is an interesting future work to extend the results in this paper to the more complicated settings, especially in the asymmetric setting.}

Some previous work~\cite{tang2015cache,dinh2016extending} achieves the linear \depth{} in several problems.
We note that they assume a much stronger model to guarantee the sequential and parallel execution order, so their algorithms need specially designed schedulers~\cite{dinh2016extending,chowdhury2017provably}.
Our algorithms are much simpler and under the nested-parallel model.
Also, all algorithms in this paper are race-free, while previous algorithms heavily rely on concurrent-writes to improve span.
The space issue of algorithms is discussed in Section~\ref{sec:para-sym}.

The dynamic programming recurrences discussed in this paper have non-local dependencies (definition given in~\cite{galil1994parallel}), and we point out that they are pretty different from the problems like edit distance or stencil computations (e.g., \cite{chowdhury2010cacheb,frigo2005cache,hirschberg1975linear,landau1986introducing}) that only have local dependencies.
We did not consider other types of dynamic programming approaches like rank convergence or hybrid $r$-way DAC algorithms~\cite{maleki2016efficient,maleki2016low,chowdhury2008cache} that cannot guarantee processor- and cache-obliviousness simultaneously.

\section{$k$-d Grid Computation Structure}\label{sec:kdcomp}

The \ourcompfull{} (short for the \ourcomp{}) is defined as a $k$-dimensional grid $C$ of size $n_1\times n_2\times \cdots\times n_k$.
Here we consider $k$ to be a small constant greater than 1.
This computation requires $k-1$ input arrays $I_1,\cdots,I_{k-1}$ and generates one output array $O$.
Each array has dimension $k-1$ and is the projection of the grid removing one of the dimensions.
Each \defn{cell} in the grid represents some certain computation that requires $k-1$ inputs and generates a temporary value.
This temporary value is ``added'' to the corresponding location in the output array using an associative operation $\oplus$. 
The $k-1$ inputs of this cell are the projections of this cell removing each (but not the last) dimensions, and the output is the projection removing the last dimension.
They are referred to as the input and output \defn{entries} of this cell.
Figure~\ref{fig:grid} illustrates such a computation in 2 and 3 dimensions.
This structure (mostly the special case for 3d as defined below) is used implicitly and explicitly by Hong and Kung~\cite{HK81} and some subsequence works (e.g.,~\cite{BallardDHS10,AggarwalCS90,IronyTT04,BallardDHS11,chowdhury2010cache}).
In this paper, we will use it as a building block to prove lower bounds and design new algorithms for dynamic programming problems.
When showing the cache complexity, we assume the input and output entries must be in the cache when computing each cell.

We refer to a \ourcompfull{} as a \defn{\sqrfull{}} (short for a \sqr{}) of size $n$ if it has size $n_1=\cdots=n_k=n$.
More concisely, we say a \ourcomp{} has size $n$ if it is square and of size $n$.

A formal definition of a square \kdcomp{3} of size $n$ is as follows:
$$O_{i,j}=\sum_k{g\left((I_{1})_{i,k}(I_{2})_{k,j},i,j,k\right)}$$
where $1\le i,j,k\le n$.
$g(\cdot)$ computes a value based on the two inputs $(I_{1})_{i,k}$ and $(I_{2})_{k,j}$ the indices, and some constant amount of data that is associated to the indices.
We assume that computing $g(\cdot)$ takes unit cost.
Each application of $g(\cdot)$ corresponds to a cell, and $(I_{1})_{i,k}$, $(I_{2})_{k,j}$ and $O_{i,j}$ are entries associated with this cell.
The sum $\sum$ is based on the associative operator $\oplus.$
Similarly, the definition for the 2d case is:
$$O_{i}=\sum_j{g\left(I_{j},i,j\right)}$$
and we can extend it to non-square cases and for $k>3$ accordingly.

We allow the output $O$ to be the same array as the input array(s) $I$.   This is used for all DP recurrences.
In these algorithms, some of the cells are empty to avoid cyclic dependencies.
For example, in a \kdcomp{2}, we may want to restrict $1\le j<i$.
In these cases, a constant fraction of the grid cells are empty.
We call such a grid an \defn{$\alpha$-full} grid for some constant $0<\alpha<1$ if at least an $\alpha\pm o(1)$ fraction of the cells are non-empty.
We will show that all properties we show for a \ourcomp{} also work for the $\alpha$-full case, since the constant $\alpha$ affects neither the lower bounds nor the algorithms.

We now show some examples that can be matched to \ourcomp{s}.
Multiplying two matrices of size $n$-by-$n$ on a semiring $(\cdot,+)$ (i.e., $O_{i,j}=\sum_k{(I_{1})_{i,k}(I_{2})_{k,j}}$) exactly matches a 3d \sqr{}.
A corresponding 2d case is when computing a matrix-vector multiplication $O_i = \sum_j I_j\cdot f(i,j)$ where $f(i,j)$ does not need to be stored.
Such applications are commonly seen in dynamic programming algorithms.
For example, the widely used LWS recurrence (Section~\ref{sec:lws}) that computes $D_j=\min_{0\le i<j}\{D_i+w(i,j)\}$ is a \kdcomp{2}, and the associative operator $\oplus$ is $\min$.  In this case the input is the same array as the output.
These are the simple cases, so even without using the \ourcomp{}, the algorithms for them in the symmetric setting are already studied in~\cite{chowdhury2006cache,Frigo99}.


However, not all DP recurrences can be viewed as \ourcomp{s} straightforwardly.
As shown above, the key aspect of deciding the dimension of a computation is the number of inputs that each basic cell $g(\cdot)$ requires.
For example, when multiplying two dense tensors, although each tensor may have multiple dimensions, each multiplication operation is only based on two entries and can be written in the previous 3d form, so the computation is a \kdcomp{3}.
Another example is the RNA recurrence that computes a 2D array
\(\displaystyle D_{i,j}=\min_{\substack{0\le p<i,0\le q<j}}\{D_{p,q}+w(p,q,i,j)\}\).
Assuming $w(p,q,i,j)$ can be queried on-the-fly, the computation is the simplest \kdcomp{2}.
Despite that the DP table has size $O(n^2)$ and $O(n^4)$ updates in total, the computation is no harder than the simplest LWS recurrence mentioned in the previous paragraph.
Similarly, in the GAP recurrence in Section~\ref{sec:gap}, each element in the DP table is computed using many other elements similar to matrix multiplication.
However, each update only requires the value of one input element and can be represented by a set of \kdcomp{2}s, unlike matrix multiplication that is a \kdcomp{3} and uses the values of two input elements in each update.
The exact correspondence between the \ourcomp{} and the DP recurrences are usually more sophisticated than their classic applications in linear algebra problems, as shown in Section~\ref{sec:dp},~\ref{sec:numerical} and appendices.
The cache-oblivious algorithms discussed in this paper are based on \ourcomp{s} with $k=2$ or $3$, but we can also find applications with larger $k$ (e.g., a Nim game with some certain rules on multiple piles~\cite{bouton1901nim}).

We note that our definition of the \ourcomp{} cannot abstract all possible recurrences and computations, but it is sufficient to analyze the DP recurrences and algorithms shown in Section~\ref{sec:dp},~\ref{sec:numerical} and appendices.
Also the \ourcomp{} is designed to analyze computations with non-local dependencies~\cite{galil1994parallel}, so it is not useful to problems such as the classic edit distance and matrix addition.


\section{Lower Bounds}\label{sec:lower}

We first discuss the lower bounds of the cache complexities for a \ourcompfull{}, which sets the target to design the algorithms in the following sections.
In Section~\ref{sec:trivial-lower-bound} we show the symmetric cache complexity.
This is a direct extension of the classic result by Hong and Kong~\cite{HK81} to an arbitrary dimension.
Then in Section~\ref{sec:asym-lower-bound} we discuss the asymmetric cache complexity when writes are more expensive than reads, which is more interesting and has a more involved analyses.

\subsection{Symmetric Cache Complexity}\label{sec:trivial-lower-bound}

The symmetric cache complexity of a \ourcomp{} is simple to analyze, yielding to the following result:
\begin{theorem}[\cite{HK81}]\label{thm:sym-lower}
The symmetric cache complexity of a \ourcompfull{} with size $n$ is $\displaystyle\Omega\left(\frac{n^k}{M^{1/(k-1)}B}\right)$.
\end{theorem}
The following proof is an extension of the proof by Hong and Kong~\cite{HK81} and we show it here again for the completeness and to help to understand the proof in Section~\ref{sec:asym-lower-bound} which has a similar outline.
\begin{proof}
In a \ourcompfull{} with size $n$ there are $n^k$ cells.
Let's sequentialize these cells in a list and consider each block of cells that considers $S=2^kM^{{k}/(k-1)}$ consecutive cells in the list.
The number of input entries required of each block is the projection of all cells in this block along one of the first $k-1$ dimensions (see Figure~\ref{fig:grid}), and this is similar for the output.
Loomis-Whitney inequality~\cite{loomis1949inequality,BallardDHS11} indicates that the overall number of input and output entries is minimized when the cells are in a square $k$-d cuboid, giving a total of  $k S^{(k-1)/k} = k\cdot \left(2M^{1/(k-1)}\right)^{k-1}\ge 4M$ input and output entries.
Since only a total of $M$ entries can be held in the cache at the beginning of the computation of this block, the number of cache-line transfer for the input/output during the computation for such a block is $\Omega(M/B)$.
Since there are $n^k/S=\Theta(n^kM^{-k/(k-1)})$ such blocks, the cache complexity of the entire computation is $\Omega(M/B)\cdot n^k/S=\Omega(n^k/(M^{1/(k-1)}B))$.
\end{proof}

Notice that the proof does not assume cache-obliviousness, but the lower bound is asymptotically tight by applying a sequential cache-oblivious algorithm that is based on $2^k$-way divide-and-conquer~\cite{Frigo99}.

\subsection{Asymmetric Cache Complexity}\label{sec:asym-lower-bound}

We now consider the asymmetric cache complexity of a \ourcompfull{} assuming writes are more expensive.
Unfortunately, this case is significantly harder than the symmetric setting.
Again for simplicity we first analyze the \sqr{} of size $n$, which can be extended to the more general cases similar to~\cite{Frigo99}.

Interestingly, there is no specific pattern that a cache-oblivious algorithm has to follow.
Some existing algorithms use ``buffers'' to support cache-obliviousness (e.g., \cite{Arge03}), and many others use a recursive divide-and-conquer framework.
For the recursive approaches, when the cache complexity of the computation is not leaf-dominated (like various sorting algorithms~\cite{Frigo99,BFGGS15}), a larger fan-out in the recursion is more preferable (usually set to $O(\sqrt{n})$).
Otherwise, when it is leaf-dominated, existing efficient algorithms all pick a constant fan-out in the recursion in order to reach the base case and fit in the cache with maximal possible subproblem size.
All problems we discuss in this paper are in this category, so we make our analysis under the following constraints.
More discussion about this constraint in given in Section~\ref{sec:conclusion}.

\begin{definition} [CBCO paradigm]
We say a divide-and-conquer algorithm is under the \emph{constant-branching cache-oblivious (CBCO) paradigm} if it has an input-value independent computational DAG, such that each task has constant\footnote{It can exponentially depend on $k$ since we assume $k$ is a constant.} fan-outs of its recursive subtasks until the base cases, and the partition of each task is decided by the ratio of the ranges in all dimensions of the (sub)problem and independent of the cache parameters ($M$ and $B$).   
\end{definition}

Notice that $\wcost{}$ is a parameter of the main memory, instead of a cache parameter, so the  algorithms can be aware of it.  One can define resource-obliviousness~\cite{cole2010resource} so that the value of $\wcost{}$ is not exposed to the algorithms, but this is out of the scope of this paper.

\hide{
For divide-and-conquer CO algorithms, the overall memory footprints in different recursive levels can either be root-dominated, balanced, or leaf-dominated.
For existing CO-MM algorithms on symmetric memories that have optimal cache complexity, each task has a constant number of subtasks.  This is because the recurrence of the I/O cost in matrix multiplication is leaf-dominated, and constant branching can fit the subtasks into the cache when the memory footprint is between $M/c'$ and $M$ for $c'$ is no more than the number of branches.  A non-constant branching on the other hand, can delay this in the worst case, and thus increases the overall cache complexity.  Notice that the CBCO paradigm is not required for root-dominated or balanced divide-and-conquer computation.
}

We now prove the (sequential) lower bound on the asymmetric cache complexity of a \ourcomp{} under the CBCO paradigm.
The constant branching and the partition based on the ratio of the ranges in all dimensions restrict the computation pattern and lead to the ``scale-free'' property of the cache-oblivious algorithms: the structure or the ``shape'' of each subproblem in the recursive levels is similar, and only the size varies.
The proof references this property when it is used.
The CBCO paradigm also restricts the shape of the computation, which is a stronger assumption than the Loomis-Whitney inequality used in the previous proof.

\begin{theorem}\label{thm:asym-lower}
The asymmetric cache complexity of \ourcomp{} is $\displaystyle \Omega\left(n^k\wcost^{1/k}\over M^{1/(k-1)}B\right)$ under the CBCO paradigm.
\end{theorem}
\begin{proof}
We prove the lower bound using the same approach in Section~\ref{sec:trivial-lower-bound}---putting all operations (cells) executed by the algorithm in a list and analyzing blocks of $S$ cells.
The cache can hold $M$ entries as temporary space for the computation.
For the lower bound, we only consider the computation in each cell without considering the step of adding the calculated value back into the output array, which only makes the problem easier.
Again when applying the computation of each cell, the $k$ input and output entries have to be in the cache.

For a block of cells with size $S$, the cache needs to hold the entries in $I_1,\cdots,I_{k-1}$ and $O$ corresponding to the cells in this block at least once during the computation.
Similar to the symmetric setting discussed above, the number of entries is minimized when the sequence of operations are within a $k$-d cuboid of size $S=a_1\times a_2\times \cdots \times a_k$ where the projections on $I_i$ and $O$ are $(k-1)$-d arrays with sizes $a_1\times \cdots \times a_{i-1} \times a_{i+1} \times \cdots \times a_k$ and $a_1\times\cdots\times a_{k-1}$.
Namely, the number of entries is at least $S/B\cdot 1/a_i$ for the corresponding input or output array.

Note that the input arrays are symmetric to each other regarding the access cost, but in the asymmetric setting storing the output entries is more expensive since they have to be written back to the asymmetric memory.
As a result, the cache complexity is minimized when $a_1=\cdots=a_{k-1}=a$, and let's denote $a_k=ar$ where $r$ is the ratio between $a_k$ and other $a_i$.
Here we assume $r\ge 1$ since reads are cheaper.
Due to the scale-free property that $M$ and $n$ are arbitrary, $r$ should be fixed (within a small constant range) for the entire recursion.

Similar to the analysis for Theorem~\ref{thm:sym-lower}, for a block of size $S$, the read transfers required by the cache is $\displaystyle\Omega\left({n^k\over SB}\cdot \max\{a^{k-1}r-M,0\}\right)$, where $n^k/S$ is the number of such blocks, and $\max\{a^{k-1}r-M,0\}/B$ lower bounds the number of reads per block because at most $M$ entries can be stored in the cache from the previous block.
Similarly, the write cost is $\displaystyle\Omega\left({\wcost{}n^k\over SB}\cdot \max\{a^{k-1}-M,0\}\right)$.
In total, the cost is:
\begin{eqnarray*}
  Q&=&\Omega\left({n^k\over SB}\cdot \left(\max\{a^{k-1}r-M,0\}+\wcost{}\max\{a^{k-1}-M,0\}\right)\right)\\
  &=&\Omega\left({n^k\over SB}\left(\max\{S^{(k-1)/k}r^{1/k}-M,0\}+\wcost{}\max\left\{{S^{(k-1)/k}\over r^{(k-1)/k}}-M,0\right\}\right)\right)
\end{eqnarray*}
The second step is due to $S=\Theta(a^kr)$.

The cost decreases as the increase of $S$, but it has two discontinuous points $S_1=M^{k/(k-1)}/r^{1/(k-1)}$ and $S_2=M^{k/(k-1)}r$.
Therefore,
\begin{eqnarray*}
  Q&=&\Omega\left({n^k\over S_1B}S_1^{(k-1)/k}r^{1/k}+{n^k\over S_2B}\left(S_2^{(k-1)/k}r^{1/k}+{\wcost{} S_2^{(k-1)/k}\over r^{(k-1)/k}}\right)\right)\\
  &=&\Omega\left({n^k\over S_1^{1/k}B}r^{1/k}+{n^k\over S_2^{1/k}B}\left(r^{1/k}+{\wcost{} \over r^{(k-1)/k}}\right)\right)  \\
  &=&\Omega\left({n^k\over M^{1/k}B}\left(r^{1/k}+{\wcost{}\over r}\right)\right)
\end{eqnarray*}
Setting $r=\wcost^{(k-1)/k}$ minimizes ${n^k\over M^{1/k}B}\left(r^{1/k}+{\wcost{}\over r}\right)$.
In this case, the lower bound of the asymmetric cache complexity $Q$ is  $\displaystyle\Omega\left(n^k\wcost^{1/k}\over M^{1/(k-1)}B\right)$, and this leads to the theorem.
\end{proof}


\section{A Matching Upper Bound on Asymmetric Memory}\label{sec:algorithm}

In the sequential and symmetric setting, the classic cache-oblivious divide-and-conquer algorithms to compute the \ourcomp{} (e.g., 3D case shown in~\cite{Frigo99}) is optimal.
In the asymmetric setting, the proof of Theorem~\ref{thm:asym-lower} indicates that the classic algorithm is not optimal and off by a factor of $\wcost{}^{(k-1)/k}$.
This gap is captured by the balancing factor $r$ in the proof, which leads to more cheap reads and less expensive writes in each sub-computation.

We now show that the lower bound in Theorem~\ref{thm:asym-lower} is tight by a (sequential) cache-oblivious algorithm with such asymmetric cache complexity.
The algorithm is given in Algorithm~\ref{algo:comp}, which can be viewed as a variant of the classic approach with minor modifications on how to partition the computation.
Notice that in line~\ref{line:partition1} and~\ref{line:sym-start}, ``conceptually'' means the partitions are used for the ease of algorithm description.
In practice, we can just pass the ranges of indices of the subtask in the recursion, instead of actually partitioning the arrays.

\begin{algorithm}[t]
\caption{$\mf{Asym-Alg}(I_1,\cdots,I_{k-1},O)$}
\label{algo:comp}
\KwIn{$k-1$ input arrays $I_1,\cdots,I_{k-1}$, read/write asymmetry $\wcost{}$}
\KwOut{Output array $O$}
The computation has size $n_1\times n_2\times \cdots\times n_k$
    \vspace{0.5em}

    \uIf {$I_1,\cdots,I_{k-1},O$ are small enough} {
        Solve the base case and \Return
    }
    $i\gets \argmax_{1\le i\le k}\{n_ix_i\}$ where $x_k=\wcost^{-(k-1)/k}$ and $x_j=1$ for $1\le j<k$\label{line:pick-dim}\\
    \uIf {$i=k$\label{line:asym}} {
        (Conceptually) equally partition $I_1,\cdots,I_{k-1}$ into $\{I_{1,a},I_{1,b}\},\cdots,\{I_{k-1,a},I_{k-1,b}\}$ on $k$-th dimension\label{line:partition1}\\
        $\mf{Asym-Alg}(I_{1,a},\cdots,I_{k-1,a},O)$\\
        $\mf{Asym-Alg}(I_{1,b},\cdots,I_{k-1,b},O)$\label{line:asym-end}\\
    } \uElse { 
        (Conceptually) equally partition $I_1,\cdots,I_{i-1},I_{i+1},\cdots,I_{k-1},O$ into $\{I_{1,a},I_{1,b}\},\cdots,\{I_{k-1,a},I_{k-1,b}\},\{O_a,O_b\}$ on $i$-th dimension\label{line:sym-start}\\
        $\mf{Asym-Alg}(I_{1,a},\cdots,I_{i-1,a},I_{i},I_{i+1,a},\cdots,I_{k-1,a},O_a)$\\
        $\mf{Asym-Alg}(I_{1,b},\cdots,I_{i-1,b},I_{i},I_{i+1,b},\cdots,I_{k-1,b},O_b)$\label{line:sym-end}
    }
\end{algorithm}

Compared to the classic approaches (e.g., \cite{Frigo99}) that partition the largest input dimension among $n_i$, the only underlying difference in the new algorithm is in line~\ref{line:pick-dim}---when partitioning the dimension not related to the output array $O$ (line~\ref{line:partition1}--\ref{line:asym-end}), $n_k$ has to be $\wcost{}^{(k-1)/k}$ times larger than $n_1,\cdots,n_{k-1}$.
This modification introduces an asymmetry between the input size and output size of each subtask, which leads to fewer writes in total and an improvement in the cache efficiency.

For simplicity, we show the asymmetric cache complexity for square grids (i.e., $n_1=\cdots=n_k$) and $n=\Omega(\wcost^{(k-1)/k}M)$, and the general case can be analyzed similar to~\cite{Frigo99}.

\begin{theorem}
Algorithm~\ref{algo:comp} computes the \ourcomp{} of size $n$ with asymmetric cache complexity $\displaystyle\Theta\left(n^k\wcost^{1/k}\over M^{1/(k-1)}B\right)$.
\end{theorem}
\begin{proof}
We separately analyze the numbers of reads and writes in Algorithm~\ref{algo:comp}.
In the sequential execution of Algorithm~\ref{algo:comp}, each recursive function call only requires $O(1)$ extra temporary space. 
Also, our analysis ignores rounding issues since they will not affect the asymptotic bounds.

When starting from the \sqr{} at the beginning, the algorithm first partitions in the first $k-1$ dimensions (via line~\ref{line:sym-start} to \ref{line:sym-end}) into $\wcost^{(k-1)^2/k}$ subproblems (referred to as \emph{second-phase} subproblems) each with size $(n/\wcost^{(k-1)/k})\times \cdots\times (n/\wcost^{(k-1)/k})\times n$, and then partition $k$ dimensions in turn until the base case is reached.

The number of writes of the algorithm $W(n)$ (to array $O$) follows the recurrences:
$$W'(n)=2^k W'(n/2)+O(1)$$
$$W(n)=(\wcost^{(k-1)/k})^{k-1}\cdot \left(W'(n/\wcost^{(k-1)/k})+O(1)\right)$$
where $W'(n)$ is the number of writes of the second-phase subproblems with the size of $O$ being $n\times\cdots\times n$.
The base case is when $W'(M^{1/(k-1)})=O(M/B)$.
Solving the recurrences gives $\displaystyle W'(n/\wcost^{(k-1)/k})=O\left(\frac{n^k\wcost^{1-k}}{M^{1/(k-1)}B}\right)$, and $\displaystyle W(n)=O\left(\frac{n^k\wcost^{(1-k)/k}}{M^{1/(k-1)}B}\right)$.

We can analyze the reads similarly by defining $R(n)$ and $R'(n)$.
The recurrences are therefore:
$$R'(n)=2^kR'(n/2)+O(1)$$
and
$$R(n)=(\wcost^{(k-1)/k})^{k-1}\cdot \left(R'(n/\wcost^{(k-1)/k})+O(1)\right)$$
The difference from the write cost is in the base case since the input fits into the cache sooner when $n=M^{1/(k-1)}/\wcost^{1/k}$.
Namely, $R'(M^{1/(k-1)}/\wcost^{1/k})=O(M/B)$.
By solving the recurrences, we have $\displaystyle R'(n/\wcost^{(k-1)/k})=O\left(\frac{n^k\wcost{}^{2-k}}{M^{1/(k-1)}B}\right)$ and $\displaystyle R(n)=O\left(\frac{n^k\wcost^{1/k}}{M^{1/(k-1)}B}\right)$.

The overall (sequential) asymmetric cache complexity for Algorithm~\ref{algo:comp} is:
$$Q(n)=R(n)+\wcost{}W(n)=O\left(n^k\wcost^{1/k}\over M^{1/(k-1)}B\right)$$
and combining with the lower bound of Theorem~\ref{thm:asym-lower} proves the theorem.
\end{proof}

Comparing to the classic approach, the new algorithm improves the asymmetric cache complexity by a factor of $O(\wcost{}^{(k-1)/k})$, since the classic algorithm requires $\Theta(n^k/(M^{1/(k-1)}B))$ reads and writes.
Again here we assume $n^{k-1}$ is much larger than $M$.
Otherwise, the lower and upper bounds should include $\Theta(\wcost{}n^{k-1}/B)$ for storing the output $O$ on the memory.

\hide{
\begin{corollary}
When the size of the matrix multiplication is $n\times m\times k$, Algorithm~\ref{algo:AsymMM} has cache complexity $$\displaystyle\Theta\left(mk\wcost^{1/3}+nm\wcost^{1/3}+nk\wcost+{nmk\wcost^{1/3}\over B\min\{n,m/\wcost^{2/3},k,\sqrt{M}\}}\right)$$
\end{corollary}

To be verified.
}

\section{Parallelism}\label{sec:para}

We now show the parallelism in computing the \ourcomp{s}.
The parallel versions of the cache-oblivious algorithms only have polylogarithmic \depth{}, indicating that they are highly parallelized.

\subsection{The Symmetric Case}\label{sec:para-sym}

We first discuss how to parallelize the classic algorithm on symmetric memory.
For a \sqr{}, the algorithm partitions the $k$-dimensions in turn until the base case is reached.

Notice that in every $k$ consecutive partitions, there are no dependencies in $k-1$ of them, so we can fully parallelize these levels with no additional cost.
The only exception is during the partition in the $k$-th dimension, whereas both subtasks share the same output array $O$ and cause write concurrence.
If such two subtasks are sequentialized (like in~\cite{Frigo99}), the \depth{} is $D(n)=2D(n/2)+O(1)=O(n)$.

We now introduce the algorithm with logarithmic depth.
As just explained, to avoid the two subtasks from modifying the same elements in the output array $O$, our algorithm works as follows when partitioning the $k$-th dimension:
\begin{enumerate}
  \item Allocating two stack-allocated temporary arrays with the same size of the output array $O$ before the two recursive function calls.
  \item Applying computation for the \ourcomp{} in two subtasks using different output arrays that are just allocated (no concurrency to the other subtask).
  \item When both subtasks finish, the computed values are merged (added) back in parallel, with work proportional to the output size and $O(\log n)$ \depth{}.
  \item Deallocating the temporary arrays.
\end{enumerate}

Notice that the algorithm also works if we only allocate temporary space for one of the subtasks, while the other subtask still works on the original space for the output array.
This can be a possible improvement in practice, but in high dimensional case ($k>2$) it requires complicated details to pass the pointers of the output arrays to descendant nodes, aligning arrays to cache lines, etc.
Theoretically, this version does not change the bounds except for the stack space in Lemma~\ref{lemma:co-seq-space} when $k=2$.

We first analyze the cost of \sqr{s} of size $n$ in the \textbf{symmetric} setting, and will discuss the asymmetric setting later.

\begin{lemma}\label{lemma:co-seq-space}
The overall stack space for a subtask of size $n$ is $O(n^{k-1})$.
\end{lemma}
\begin{proof}
The parallel algorithm allocates memory only when partitioning the output ($k$-th) dimension.
In this case, it allocates and computes two subtasks of size $n/2$ where $n$ is the size of the output dimension.
This leads to the following recurrence:
\[S(n)=2S(n/2)+O(n^{k-1})\]
The recurrence solves to $S(n)=O(n^{k-1})$ when $k>2$ since the recurrence is root-dominated.
When $k=2$, we can apply the version that only allocates temporary space for one subtask, which decreases the constant before $S(n/2)$ to 1, and yields $S(n)=O(n)$.
Note that we only need to analyze one of the branches, since the temporary spaces that are not allocated in the direct ancestor of this subtask have already been deallocated, and will be reused for later computations for the current branch.
\end{proof}
With the lemma, we have the following corollary:
\begin{corollary}
A subtask of size $n\le M^{1/(k-1)}$ can be computed within a cache of size $O(M)$.
\end{corollary}

This corollary indicates that this modified parallel algorithm has the same sequential cache complexity since it fits into the cache in the same level as the classic algorithm (the only minor difference is the required cache size increases by a small constant factor).
Therefore we can apply the a similar analysis in~\cite{Frigo99} ($k=3$ in the paper) to show the following lemma:
\begin{lemma}
The sequential symmetric cache complexity of the parallel cache-oblivious algorithm to compute a \ourcomp{} of size $n$ is $O(n^k/M^{1/(k-1)}B)$.
\end{lemma}

Assuming that we can allocate a chunk of memory in constant time, the \depth{} of this approach is simply $O(\log^2 n)$---$O(\log n)$ levels of recursion, each with $O(\log n)$ \depth{} for the additions~\cite{blelloch2010low}.

We have shown the parallel \depth{} and symmetric cache complexity.
By applying the scheduling theorem in Section~\ref{sec:prelim}, we have the following result for parallel symmetric cache complexity.
\begin{corollary}
The \ourcomp{} of size $n$ can be computed with the parallel symmetric cache complexity of $O(n^k/M^{1/(k-1)}B+pM\log^2n)$ with private caches, or $O(n^k/M^{1/(k-1)}B)$ with a share cache of size $M+pB\log^2n$.
\end{corollary}
\hide{
\begin{proof}
The total number of steals is upper bounded by $O(pD)$~\cite{Acar02}.
For private-cache case, the extra I/O cost by a stolen task is no more than the cost to refill the cache from the original thread, which is $O(M)$.
Plugging in the \depth{} gives the stated parallel cache complexity.
(Discuss share cache case.)
\end{proof}}

We now analyze the overall space requirement for this algorithm.
Lemma~\ref{lemma:co-seq-space} shows that the extra space required is $S_1=O(n^{k-1})$ for sequentially running the parallel algorithm.
Na\"ively the parallel space requirement is $pS_1$, which can be very large.
We now show a better upper bound for the extra space.
\begin{lemma}\label{lemma:sym-space}
The overall space requirement of the parallel algorithm to compute the \ourcomp{} is $O(p^{1/k}n^{k-1})$.
\end{lemma}
\begin{proof}
We analyze the total space allocated for all processors.
Lemma~\ref{lemma:co-seq-space} indicates that if the root of the computation on one processor has the output array of size $(n')^{k-1}$, then the space requirement for this task is $O((n')^{k-1})$.
There are in total $p$ processors.
There can be at most $2^k$ processors starting with their computations of size $n^{k-1}/2^{k-1}$, $(2^k)^2$ of size $n^{k-1}/(2^{k-1})^2$, until $(2^k)^q$ processors of size $n^{k-1}/(2^{k-1})^q$ where $q=\log_{2^k} p$.
This case maximizes the overall space requirement for $p$ processors, which is:
\begin{align*}
    \sum_{h=1}^{\log_{2^k} p}{O\left(\frac{n^{k-1}}{(2^{k-1})^h}\right)\cdot (2^k)^h}
    = p\cdot O\left(\frac{n^{k-1}}{(2^{k-1})^{\log_{2^k}p}}\right)=O(p^{1/k}n^{k-1})
\end{align*}
This shows the stated bound.
\end{proof}

Combining all results gives the following theorem:
\begin{theorem}\label{thm:sym-cost}
There exists a cache-oblivious algorithm to compute a \ourcomp{} of size $n$ that requires $\Theta(n^k)$ work, $\displaystyle \Theta\left(n^k\over M^{1/(k-1)}B\right)$ symmetric cache complexity, $O(\log^2 n)$ \depth{}, and $O(p^{1/k}n^{k-1})$ main memory size.
\end{theorem}

\myparagraph{Additional space required}
The following discussion is purely on the practical side and does not affect the theoretical analysis of all the theorems in this paper.

We believe the space requirement for the parallel cache-oblivious algorithm is acceptable since it is asymptotically the same as the most intuitively (non-cache-oblivious) parallel algorithm that partitions the computation into $p$ square subtasks each with size $n/p^{1/k}$.
In practice nowadays it is easy to fit several terabyte main memory onto a single powerful machine such that the space requirement can usually be satisfied.
For example, a Dell PowerEdge R940 has about $p=100$ and the main memory can hold more than $10^{12}$ integers, while the new NVRAMs will have even more capacity (up to 512GB per DIMM).
On such machines, when $k=2$, the grid needs to contain more than $10^{22}$ cells to exceed the memory size---such computation takes too long to run on a single shared-memory machine.
For $k=3$, we need about $10^{17}$ cells to exceed the main memory size, which will take weeks to execute on a highest-end shared-memory machine.
Hence, throughout the paper we focus on cache complexity and span.
Even if one wants to run such a computation, we can use the following approach to slightly change the algorithm to bound the extra space as a practical fix.

We can first partition the input dimensions for $\log_2 p$ rounds to bound the largest possible output size to be $O(n^{k-1}/p)$ (similar to the case discussed in Section~\ref{sec:para-asym}).
Then the overall extra space for all $p$ processors is limited to $O(n^{k-1})$, the same as the input/output size.
If needed, the constant in the big-O can also be bounded.
Such a change will not affect the cache complexity and the \depth{} as long as \emph{the main memory size is larger than $pM$} where $M$ is the cache size.
This is because the changes of partition order do not affect the recurrence depth, and the I/O cost is still dominated by when the subproblems fitting the cache.
In practice, DRAM size is \emph{always} several orders of magnitude larger than $pM$.

\subsection{The Asymmetric Case}\label{sec:para-asym}

Algorithm~\ref{algo:comp} considers the write-read asymmetry, which involves some minor changes to the classic cache-oblivious algorithm.
Regarding parallelism, the changes in Algorithm~\ref{algo:comp} only affect the order of the partitioning of the $k$-d grid in the recurrence, but not the parallel version and the analysis in Section~\ref{sec:para-sym}.
As a result, the \depth{} of the parallel variant of Algorithm~\ref{algo:comp} is also $O(\log^2 n)$.
The extra space upper bound is actually reduced, because the asymmetric algorithm has a higher priority in partitioning the input dimensions that does not requires allocation temporary space.
\begin{lemma}\label{lemma:asym-space}
The space requirement of Algorithm~\ref{algo:comp} on $p$ processors is $O(n^{k-1}(1+p^{1/k}/\wcost^{(k-1)/k}))$.
\end{lemma}
\begin{proof}
Algorithm~\ref{algo:comp} first partition the input dimensions until $q=O(\wcost{}^{(k-1)^2/k})$ subtasks are generated.
Then the algorithm will partition $k$ dimensions in turn.
If $p<q$, then each processor requires no more than $O(n^{k-1}/q)$ extra space at any time, so the overall extra space is $O(p\cdot n^{k-1}/q)=O(n)$.
Otherwise, the worst case appears when $O(p/q)$ processors work on each of the subtasks.
Based on Lemma~\ref{lemma:sym-space}, the extra space is bounded by $O((p/q)^{1/k}\cdot q\cdot n^{k-1}/q)=O(p^{1/k}n^{k-1}/\wcost^{(k-1)/k})$.
Combining the two cases gives the stated bounds.
\end{proof}

Lemma~\ref{lemma:asym-space} indicates that Algorithm~\ref{algo:comp} requires extra space no more than the input/output size asymptotically when $p=O(\wcost^{k-1})$, which should always be true in practice.

The challenge arises in scheduling this computation.
The scheduling theorem for the asymmetric case~\cite{BBFGGMS16} constraints on the non-leaf stack memory to be a constant size.
This contradicts the parallel version in Section~\ref{sec:para-sym}.
This problem can be fixed based on Lemma~\ref{lemma:co-seq-space} that upper bounds the overall extra memory on one task.
Therefore the stack-allocated array can be moved to the heap space.
Once a task is stolen, the first allocation will annotate a chunk of memory with size order of $|O|$ where $O$ is the current output.
Then all successive heap-based memory allocation can be simulated on this chunk of memory.
In this manner, the stack memory of each node corresponding to a function call is constant, which allows us to apply the scheduling theorem in~\cite{BBFGGMS16}.

\begin{theorem}\label{thm:asym-cost}
Algorithm~\ref{algo:comp} with input size $n$ requires $\Theta(n^k)$ work, $\displaystyle \Theta\left(n^k\wcost^{1/k}\over M^{1/(k-1)}B\right)$ asymmetric cache complexity, and $O(\log^2 n)$ \depth{} to compute a \ourcomp{} of size $n$.
\end{theorem}

\section{Dynamic Programming Recurrences}\label{sec:dp}

In this section we discuss a number of new results on dynamic programming (DP).
To show lower and upper bounds on parallelism and cache efficiency in either symmetric and asymmetric setting, we focus on the specific DP recurrences instead of the problems.
We assume each update in the recurrences takes unit cost, just like the \ourcomp{} in Section~\ref{sec:kdcomp}.

The goal of this section is to show how the DP recurrences can be viewed as and decomposed into the \ourcomp{s}.
Then the lower and upper bounds discussed in Section~\ref{sec:lower} and~\ref{sec:algorithm}, as well as the analysis of parallelism in Section~\ref{sec:para}, can be easily applied to the computation of these DP recurrences.
When the dimension of the input/output is the same as the number of entries in each grid cell, then the sequential and symmetric versions of the algorithms in this section are the same as the existing ones discussed in~\cite{Frigo99,chowdhury2006cache,chowdhury2010cache,chowdhury2016autogen,tithi2015high}, but the others are new.
Also, the asymmetric versions and most parallel versions are new.
We improve the existing results on symmetric/asymmetric cache complexity, as well as parallel \depth{}.

\myparagraph{Symmetric cache complexity}
We show improved algorithms for a number of problems when the number of entries per cell differs from the dimension of input/output arrays.
Such algorithms are for the GAP recurrence, protein accordion folding, and the RNA recurrence.
We show that the previous cache bound $O(n^3/B\sqrt{M})$ for the GAP recurrence and protein accordion folding is not optimal, and we improve the bounds in Theorem~\ref{thm:gap} and~\ref{thm:paf}.
For the RNA recurrence, we show an optimal cache complexity of $\Theta(n^4/BM)$ in Theorem~\ref{thm:gap}, which improves the best existing result by $O(M^{3/4})$.

\myparagraph{Asymmetric cache complexity}
By applying the asymmetric version for the \ourcomp{} computation discussed in Section~\ref{sec:algorithm}, we show a uniform approach to provide write-efficient algorithms for all DP recurrences in this section.
We also shown the optimality of all these algorithms regarding asymmetric cache complexity, expect for the one for the GAP recurrence.

\myparagraph{Parallelism}
The parallelism of these algorithms is provided by the parallel algorithms discussed in Section~\ref{sec:para}.
Polylogarithmic \depth{} can be achieved in computing the 2-knapsack recurrence, and linear \depth{} in LWS recurrence and protein accordion folding.
The linear \depth{} for LWS can be achieved by previous work~\cite{tang2015cache,dinh2016extending}, but they are not race-free and in the nested-parallel model.
Meanwhile, our algorithms are arguably simpler.

\subsection{LWS Recurrence}\label{sec:lws}

We start with the simple example of the LWS recurrence where optimal sequential upper bound in the symmetric setting is known~\cite{chowdhury2006cache}.
We show new results for lower bounds, write-efficient cache-oblivious algorithms, and new span bound.

The LWS (least-weighted subsequence) recurrence~\cite{hirschberg1987least} is one of the most commonly-used DP recurrences in practice.
Given a real-valued function $w(i,j)$ for integers $0\le i<j\le n$ and $D_0$, for $1\le j\le n$,
$$D_j=\min_{0\le i<j}\{D_i+w(i,j)\}$$
This recurrence is widely used as a textbook algorithm to compute optimal 1D clustering~\cite{kleinberg2006algorithm}, line breaking~\cite{knuth1981breaking}, longest increasing sequence, minimum height B-tree, and many other practical algorithms in molecular biology and geology~\cite{galil1992dynamic,galil1994parallel}, computational geometry problems~\cite{aggarwal1990applications}, and more applications in~\cite{kunnemann2017fine}.
Here we assume that $w(i,j)$ can be computed in constant work based on a constant size of input associated to $i$ and $j$, which is true for all these applications.
Although different special properties of the weight function $w$ can lead to specific optimizations, the study of recurrence itself is interesting, especially regarding cache efficiency and parallelism.

We note that the computation of this recurrence is a standard \kdcomp{2}.
Each cell $g(D_i,i,j)=D_i+w(i,j)$ and updates $D_j$ as the output entry, so Theorem~\ref{thm:sym-lower} and~\ref{thm:asym-lower} show lower bounds on cache complexity on this recurrence (the grid is (1/2)-full).

We now introduce cache-oblivious implementation considering the data dependencies.
Chowdhury and Ramachandran~\cite{chowdhury2006cache} solves the recurrence with $O(n^2)$ work and $O(n^2/BM)$ symmetric cache complexity.
The algorithm is simply a divide-and-conquer approach and we describe and extend it based on \ourcomp{s}.
A task of range $(p,q)$ computes the cells $(i,j)$ such that $p\le i<j\le q$.
To compute it, the algorithm generates two equal-size subtasks $(p,r)$ and $(r+1,q)$ where $r=(p+q)/2$, solves the first subtask $(p,r)$ recursively, then computes the cells corresponding to $w(i,j)$ for $p\le i\le r<j \le q$, and lastly solves the subtask $(r+1,q)$ recursively.
Note that the middle step also matches a \kdcomp{2} with no dependencies between the cells, which can be directly solved using the algorithms in Section~\ref{sec:algorithm}.
This leads to the cache complexity and \depth{} to be:
\[Q(n)=2Q(n/2)+Q_{2C}(n/2)\]
\[D(n)=2D(n/2)+D_{2C}(n/2)\]
Here $2C$ denotes the computation of a \kdcomp{2}.
The recurrence is root-dominated with base cases $Q(M)=\Theta(M/B)$ and $D(1)=1$.
This solves to the following theorem.

\begin{theorem}\label{thm:lws}
The LWS recurrence can be computed in $\Theta(n^2)$ work, $\displaystyle \Theta\left(\frac{n^2}{BM}\right)$ and $\displaystyle \Theta\left(\frac{\wcost{}^{1/2} n^2}{BM}\right)$ optimal  symmetric and asymmetric cache complexity respectively, and $O(n)$ \depth{}.
\end{theorem}

\hide{
Notice that Galil and Park~\cite{galil1994parallel} denoted that a task $(p,q)$ can be solved in $O(n'^3)$ work and $O(\log^2 n')$ \depth{} for $n'=q-p$, based on recursive matrix multiplication.
Since matrix multiplication is a \kdcomp{3}, plugging in the result from Section~\ref{sec:mm} gives the cache complexity as $O(n'^2(1+n'/\sqrt{M})/B)$.
Thus the cache-oblivious algorithm to compute LWS recurrence can switch to this matrix-multiplication-based approach when a task has size less or equal to $t=O(\sqrt{N})$, which does not include extra work, decreases the \depth{}, but may potentially increase the I/O cost.
More specifically, we have:
\begin{corollary}
LWS recurrence can be computed in $\Theta(n^2)$ work, $\displaystyle \Theta\left(\frac{nt}{B}+\frac{nt^2}{B\sqrt{M}}+\frac{n^2}{BM}\right)$ and $\displaystyle \Theta\left(\frac{\wcost{}nt}{B}+\frac{\wcost^{1/3}nt^2}{B\sqrt{M}}+\frac{\wcost^{1/2}n^2}{BM}\right)$ cache complexity in symmetric and asymmetric memory respectively, and $O(n/t\cdot\log^2 n)$ and $O(\wcost{}n/t\cdot\log^2 n)$ \depth{}, for $t=O(\sqrt{N})$.
\end{corollary}
}

\subsection{GAP Recurrence}\label{sec:gap}

We now consider the GAP recurrence that the analysis of the lower bounds and the new algorithm make use of multiple grid computation.
The GAP problem~\cite{galil1989speeding,galil1994parallel} is a generalization of the edit distance problem that
has many applications in molecular biology, geology, and speech recognition.
Given a source string $X$ and a target string $Y$, other than changing one character in the string, we can apply a sequence of consecutive deletes that corresponds to a gap in $X$, and a sequence of consecutive inserts that corresponds to a gap in $Y$.
For simplicity here we assume both strings have length $n$, but the algorithms and analyses can easily be adapted to the more general case.
Since the cost of such a gap is not necessarily equal to the sum of the costs of each individual deletion (or insertion) in that gap,
we define $w(p,q)$ $(0\le p < q\le n)$ as the cost of deleting the substring of $X$ from $(p+1)$-th to $q$-th character, $w'(p,q)$ for inserting the substring of $Y$ accordingly, and
$r(i,j)$ as the cost to change the $i$-th character in $X$ to $j$-th character in $Y$.

Let $D_{i,j}$ be the minimum cost for such transformation from the prefix of $X$ with $i$ characters to the prefix of $Y$ with $j$ characters, the recurrence for $i,j>0$ is:
\[\displaystyle D_{i,j}=\min\left\{\begin{matrix}\min_{0\le q<j}\{D_{i,q}+w'(q,j)\}\\\min_{0\le p<i}\{D_{p,j}+w(p,i)\}\\D_{i-1,j-1}+r(i,j)\end{matrix}\right.\]
corresponding to either replacing a character, inserting or deleting a substring.
The boundary is set to be $D_{0,0}=0$, $D_{0,j}=w(0,j)$ and $D_{i,0}=w'(0,i)$.
The diagonal dependency from $D_{i-1,j-1}$ will not affect the asymptotic analysis since it will at most double the memory footprint, so it will not show up in the following analysis.

The best existing algorithms on GAP Recurrence~\cite{chowdhury2006cache,tang2017} have symmetric cache complexity of $O(n^3/B\sqrt{M})$.
This upper bound seems to be reasonable, since in order to compute $D_{i,j}$, we need the input of two vectors $D_{i,q}$ and $D_{p,j}$, which is similar to matrix multiplication and other algorithms in Section~\ref{sec:numerical}.
However, as indicated in Section~\ref{sec:kdcomp},
each update in GAP only requires one entry, while matrix multiplication has two.
Therefore, if we ignore the data dependencies, the first line of the GAP recurrence can be viewed as $n$ LWS recurrences, independent of the dimension of $i$ (similarly for the second line).
This derives a lower bound on cache complexity to be that of an LWS recurrence multiplied by $2n$, which is $\Omega(n^3/BM)$ (assuming $n>M$).
Hence, the gap between the lower and upper bounds is $\Theta(\sqrt{M})$.

We now discuss an I/O-efficient algorithm to close this gap.
This algorithm is not optimal, but reduce it to $1+o(1)$.
How to remove the low-order term remains as an open problem.
The new algorithm is similar to Chowdhury and Ramachandran's approach~\cite{chowdhury2006cache} based on divide-and-conquer to compute the output $D$. 
The algorithm recursively partitions $D$ into four equal-size quadrants $D_{00}$, $D_{01}$, $D_{10}$ and $D_{11}$, and starts to compute $D_{00}$ recursively.
After this is done, it uses the computed value in $D_{00}$ to update $D_{01}$ and $D_{10}$.
Then the algorithm computes $D_{01}$ and $D_{10}$ within their own ranges, updates $D_{11}$ using the results from $D_{01}$ and $D_{10}$, and solves $D_{11}$ recursively at the end.
The high-level idea is shown in Figure~\ref{fig:gap}.

\begin{figure}[tp]
\centering
  \includegraphics[width=.8\columnwidth]{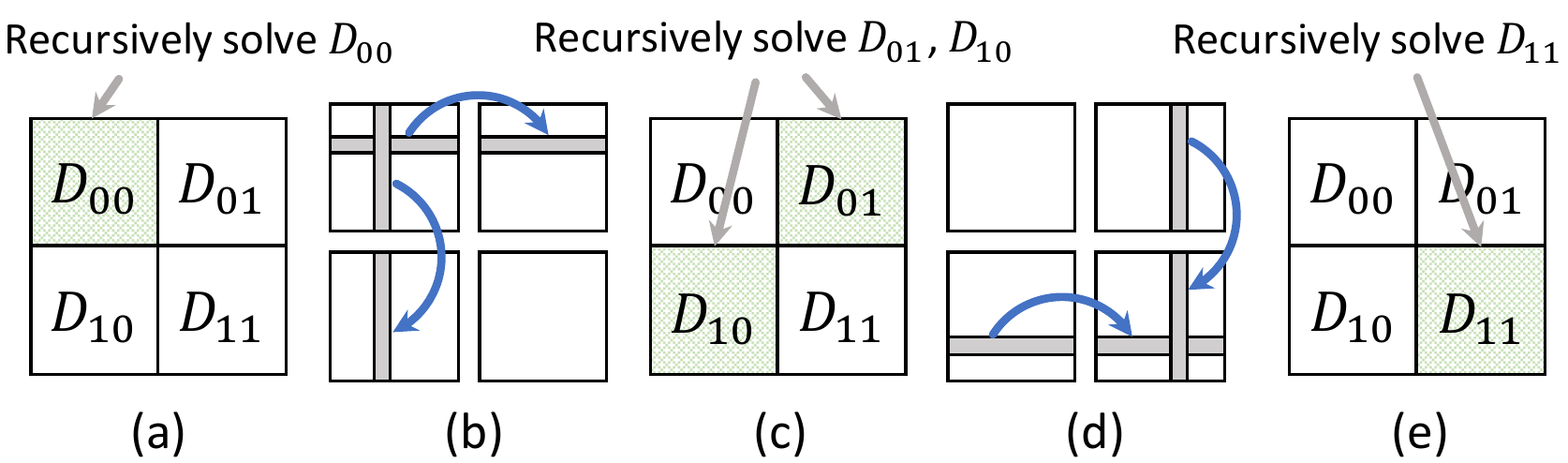}
  \caption{The new cache-oblivious algorithm for GAP recurrences ($n$ is the input size).  The algorithm has five steps.
  Step (a) first recursively solves the $D_{00}$ quadrant, then Step (b) apply $n/2$ inter-quadrant column updates and $n/2$ row updates, each corresponding to a \kdcomp{2}.
  After that, Step (c) recursively solves $D_{01}$ and $D_{10}$, Step (d) applies another $n$ inter-quadrant updates, and finally Step (e) recursively solves $D_{11}$.  More details about maintaining cache-efficiency is described in Section~\ref{sec:gap} in details.}\label{fig:gap}
\end{figure}

We note that in Steps (b) and (d), the inter-quadrant updates compute $2\times (n'/2)$ LWS recurrences (with no data dependencies) each with size $n'/2$ (assuming $D$ has size $n'\times n'$).
Therefore, our new algorithm reorganizes the data layout and the order of computation to take advantage of our I/O-efficient and parallel algorithm on \kdcomp{2}s.
Since the GAP recurrence has two independent sections one in a column and the other in a row, we keep two copies of $D$, one organized in column major and the other in row major.
Then when computing on the inter-quadrant updates as shown in Steps (b) and (d), we start $2\times (n'/2)$ parallel tasks each with size $n'/2$ and compute a \kdcomp{2} on the corresponding row or column, taking the input and output with the correct representation.
These updates require work and cache complexity shown in Theorem~\ref{thm:lws}.
We also need to keep the consistency of the two copies.
After the update of a quadrant $D_{01}$ or $D_{10}$ is finished, we apply a matrix transpose~\cite{blelloch2010low} to update the other copy of this quadrant by taking a $\min$ as the associative operator $\oplus$, so that the two copies of $D$ are consistent before Steps~(c) and~(e).
The cost of the transpose is a lower-order term.
For the quadrant $D_{11}$, we wait until the two updates from $D_{01}$ and $D_{10}$ finish, and then apply the matrix transpose to update the values in each other.
It is easy to check that by induction, the values in both copies in a quadrant are update-to-date at the beginning of each recursion in Step~(c) and~(e).  

Our new algorithm still requires $\Theta(n^3)$ work since it does not require extra asymptotic work.
The cache complexity and \depth{} satisfy:
\[Q(n)=4Q(n/2)+4(n/2)\cdot Q_{2C}(n/2)\]
\[D(n)=3D(n/2)+2D_{2C}(n/2)\]
The coefficients are easily shown by Figure~\ref{fig:gap}.
We first discuss the symmetric setting.
The base cases are $Q(\sqrt{M})=O(M/B)$ and $Q_{2C}(m)=O(m/B)$ for $m\le M$.  
This is a ``balanced'' recurrence with $O(M/B)$ I/O cost per level for $\log_2 \sqrt{M}$ levels.
This indicates $Q(M)=O((M/B)\log_2\sqrt{M})$.
The top-level computation is root dominated since the overall number of cells in a level decreases by a half after every recursion.  
Therefore, if $n>M$, $Q(n)=O(n^2Q(M)/M)+O(n)\cdot Q_{2C}(n)=O(n^2/B\cdot(n/M+\log_2 \sqrt{M}))$, which is the base-case cost plus the top-level cost.  
Otherwise, all input/output for each 2d grid in the inter-quadrant update fit in the cache, so we just need to pay $O(n^2/B)$ I/O cost for $\log_2(n/\sqrt{M})$ rounds of recursion, leading to $Q(n)=O(n^2\log_2(n/\sqrt{M})/B)$.
Similarly we can show the asymmetric results by plugging in different base cases.

\begin{theorem}\label{thm:gap}
The GAP recurrence can be computed in $\Theta(n^3)$ work, $O(n^{\log_23})$ \depth{}, symmetric cache complexity of $$\displaystyle O\left(\frac{n^2}{B}\cdot\left(\frac{n}{M}+\log_2\min\left\{\frac{n}{\sqrt{M}},\sqrt{M}\right\}\right)\right)$$ and asymmetric cache complexity of $$\displaystyle O\left(\frac{n^2}{B}\cdot\left(\frac{\wcost^{1/2}n}{M}+\wcost{}\log_2\min\left\{\frac{n}{\sqrt{M}},\sqrt{M}\right\}\right)\right)$$
\end{theorem}

Compared to the previous results~\cite{chowdhury2006cache,chowdhury2010cache,chowdhury2016autogen,itzhaky2016deriving,tithi2015high,tang2017},
the improvement on the symmetric cache complexity is asymptotically $O(\sqrt{M})$ (i.e., $n$ approaching infinity).  For smaller range of $n$ that $O(\sqrt{M})\le n\le O(M)$, the improvement is $O(n/\sqrt{M}/\log(n/\sqrt{M}))$.
(The computation fully fit into the cache when $n<O(\sqrt{M})$.)
\smallskip

\hide{
Similar to the LWS recurrence, Galil and Park~\cite{galil1994parallel} indicate that a subtask of size $n'$ in the recursion can be solved in a matrix-multiplication-based approach, which uses $O(n'^6)$ work, $O(\log^2 n')$ \depth{} and $O(n'^4(1+n'^2/\sqrt{M})/B)$ cache complexity using the analysis in Section~\ref{sec:algorithm}.
As a result, if the \depth{} of the previous algorithm is not satisfied, we can switch to this version when a subtask has size no more than $t=O(n^{1/4})$.
\begin{corollary}
GAP recurrence can be computed in $\Theta(n^3)$ work,
$\displaystyle \Theta\left(\frac{n^2t^2}{B}+\frac{n^2t^4}{B\sqrt{M}}+\frac{n^3}{BM}\right)$
and
$\displaystyle \Theta\left(\frac{\wcost{}n^2t^2}{B}+\frac{\wcost^{1/3}n^2t^4}{B\sqrt{M}}+\frac{\wcost^{1/2}n^3}{BM}\right)$
cache complexity in symmetric and asymmetric memory respectively, and
$\tilde O((n/t)^{\log_23})$
and
$\tilde O(\wcost{}(n/t)^{\log_23})$
\depth{}, for $t=O({n}^{1/4})$.
\end{corollary}
}

\myparagraph{Protein accordion folding}
The recurrence for protein accordion folding~\cite{tithi2015high} is $D_{i,j}=\max_{1\le k< j-1}\{D_{j-1,k}+w(i,j,k)\}$ for $1\le j<i\le n$, with $O(n^2/B)$ cost to precompute $w(i,j,k)$.
Although there are some minor differences, from the perspective of the computation structure, the recurrence can basically be viewed as only containing the first section of the GAP recurrence.
As a result, the same lower bounds of GAP can also apply to this recurrence.

In terms of the algorithm, we can compute $n$ \kdcomp{2}s with the increasing order of $j$ from $1$ to $n$, such that the input are $D_{j-1,k}$ for $1\le k<j-1$ and the output are $D_{i,j}$ for $j<i\le n$.
For short, we refer to a \kdcomp{2} as a task.
However, the input and output arrays are in different dimensions.
To handle it, we use a imilar method to the GAP algorithm that keeps two separate copies for $D$, one in column-major and one in row-major.
They are used separately to provide the input and output for the \kdcomp{2}.
We apply the transpose in a divide-and-conquer manner---once the first half of the tasks finish, we transpose all computed values from the output matrix to the input matrix (which is a square matrix), and then compute the second half of the task.
Both matrix transposes in the first and second halves are applied recursively with geometrically decreasing sizes.
The correctness of this algorithm can be verified by checking the data dependencies so that all required values are computed and moved to the correct positions before they are used for further computations.

The cache complexity is from two subroutines: the computations of \kdcomp{2}s and matrix transpose.
The cost of \kdcomp{2}s is simply upper bounded by $n$ times the cost of each task, which is $O(n^2/B\cdot(1+n/M))$ and $O(n^2/B\cdot(\wcost{}+\wcost{}^{1/2}n/M))$ for symmetric and asymmetric cache complexity, and $O(n\log^2 n)$ \depth{}.
For matrix transpose, the cost can be verified in the following recursions.
\[Q(n)=2Q(n/2)+Q_{\smb{Tr}}(n/2)\]
\[D(n)=2D(n/2)+D_{\smb{Tr}}(n/2)\]
where $\mb{Tr}$ indicates the matrix transpose.
The base case is $Q(\sqrt{M})=O(M/B)$ and $D(1)=1$.
Applying the bound for matrix transpose~\cite{blelloch2010low} provides the following theorem.
\begin{theorem}\label{thm:paf}
Protein accordion folding can be computed in $O(n^3)$ work, symmetric and asymmetric cache complexity of $\displaystyle \Theta\left(\frac{n^2}{B}\left(1+\frac{n}{M}\right)\right)$ and $\displaystyle \Theta\left(\frac{n^2}{B}\left(\wcost{}+\frac{\wcost{}^{1/2}n}{M}\right)\right)$ respectively, and $O(n\log^2 n)$ \depth{}.
\end{theorem}
The cache bounds in both symmetric and asymmetric cases are optimal with respect to the recurrence.

\subsection{RNA Recurrence}\label{sec:rna}

The RNA problem~\cite{galil1994parallel} is a generalization of the GAP problem.
In this problem a weight function $w(p,q,i,j)$ is given, which is the cost to delete the substring of $X$ from $(p+1)$-th to $i$-th character and insert the substring of $Y$ from $(q+1)$-th to $j$-th character.
Similar to GAP, let $D_{i,j}$ be the minimum cost for such transformation from the prefix of $X$ with $i$ characters to the prefix of $Y$ with $j$ characters, the recurrence for $i,j>0$ is:
\[\displaystyle D_{i,j}=\min_{\substack{0\le p<i\\0\le q<j}}\{D_{p,q}+w(p,q,i,j)\}\]
with the boundary values $D_{0,0}$, $D_{0,j}$ and $D_{i,0}$.
This recurrence is widely used in computational biology, like to compute the secondary structure of RNA~\cite{waterman1978rna}.

While the cache complexity of this recurrence seems to be hard to analyze in previous papers, it fits into the framework of this paper straightforwardly.
Since each computation in the recurrence only requires one input value, the whole recurrence can be viewed as a \kdcomp{2}, with both the input and output as $D$.
The 2d grid is (1/4)-full, so we can apply the lower bounds in Section~\ref{sec:algorithm} here.

Again for a matching upper bound, we need to consider the data dependencies.
We can apply the similar technique as in the GAP algorithm to partition the output $D$ into four quadrants, compute $D_{00}$, then $D_{01}$ and $D_{10}$, and finally $D_{11}$.
Each inter-quadrant update corresponds to a 1/2-full \kdcomp{2}.
Here maintaining two copies of the array is not necessary with the tall-cache assumption $M=\Omega(B^2)$. 
Applying the similar analysis in GAP gives the following result:
\begin{theorem}\label{thm:rna}
The RNA recurrence can be computed in $\Theta(n^4)$ work, optimal symmetric and asymmetric cache complexity of $\displaystyle \Theta\left(\frac{n^4}{BM}\right)$ and $\displaystyle \Theta\left(\frac{\wcost{}^{1/2} n^4}{BM}\right)$ respectively, and $O(n^{\log_23})$ \depth{}.
\end{theorem}


\subsection{Parenthesis Recurrence}\label{sec:parenthesis}

The Parenthesis recurrence solves the following problem: given a linear sequence of objects, an associative binary operation on those objects, and the cost of performing that operation on any two given (consecutive) objects (as well as all partial results), the goal is to compute the min-cost way to group the objects by applying the operations over the sequence.
Let $D_{i,j}$ be the minimum cost to merge the objects indexed from $i+1$ to $j$ (1-based), the recurrence for $0\le i<j\le n$ is:
\[D_{i,j}=\min_{i<k<j}\{D_{i,k}+D_{k,j}+w(i,k,j)\}\]
where $w(i,k,j)$ is the cost to merge the two partial results of objects indexed from $i+1$ to $k$ and those from $k+1$ to $j$.
Here the cost function is only decided by a constant-size input associated to indices $i$, $j$ and $k$.
$D_{i,i+1}$ is initialized, usually as 0.
The applications of this recurrence include the matrix chain product, construction of optimal binary search trees, triangulation of polygons, and many others shown in~\cite{CLRS,galil1992dynamic,galil1994parallel,yao1980efficient}.

The computation of this recurrence (without considering dependencies) is a (1/3)-full \kdcomp{3}, which has the same lower bound shown in Corollary~\ref{thm:mm}.

The divide-and-conquer algorithm that computes this recurrence is usually hard to describe (e.g., it takes several pages in~\cite{chowdhury2016autogen,itzhaky2016deriving} although they also describe their systems simultaneously).
We claim that under the view of our \ourcomp{s}, this algorithm is conceptually as simple as the other algorithms.
Again this divide-and-conquer algorithm partitions the state $D$ into quadrants, but at this time one of them ($D_{10}$) is empty since $D_{i,j}$ does not make sense when $i>j$.
The quadrant $D_{01}$ depends on the other two.
The algorithm first recursively computes $D_{00}$ and $D_{11}$, then updates $D_{01}$ using the computed values in $D_{00}$ and $D_{11}$, and finally recursively computes $D_{01}$.
Here $D_{01}$ is square, so the recursive computation of $D_{01}$ is almost identical to that in RNA or GAP recurrence (although the labeling of the quadrants is slightly changed): breaking a subtask into four quadrants, recursively solving each of them in the correct order while applying inter-quadrant updates in the middle.
The only difference is when the inter-quadrant updates are processed, each update requires two values, one in $D_{01}$ and another in $D_{00}$ or $D_{11}$.
This is the reason that Parenthesis is 3d while RNA and GAP are 2d.
The correctness of this algorithm can be shown inductively.

\begin{theorem}\label{thm:paren}
The Parenthesis recurrence can be computed in $\Theta(n^3)$ work, optimal symmetric and asymmetric cache complexity of $\displaystyle \Theta\left(\frac{n^3}{B\sqrt{M}}\right)$ and $\displaystyle \Theta\left(\frac{\wcost{}^{1/3} n^3}{B\sqrt{M}}\right)$ respectively, and $O(n^{\log_23})$ \depth{}.
\end{theorem}

\subsection{2-Knapsack Recurrence}

Given $A_{i}$ and $B_i$ for $0\le i\le n$, the 2-knapsack recurrence computes:
\[D_i=\min_{0\le j\le i}\{A_{j}+B_{i-j}+w(j,i-j,i)\}\]
for $0\le i\le n$.
The cost function $w(j,i-j,i)$ relies on constant input values related on indices $i$, $i-j$ and $j$.
To the best of our knowledge, this recurrence is first discussed in this paper.
We name is the ``2-knapsack recurrence'' since it can be interpreted as the process of finding the optimal strategy in merging two knapsacks, given the optimal local arrangement of each knapsack stored in $A$ and $B$.
Although this recurrence seems trivial, the computation structure of this recurrence actually forms some more complicated DP recurrence.
For example, many problems on trees\footnote{Such problems can be: (1) computing a size-$k$ independent vertex set on a tree that maximizes overall neighbor size, total vertex weights, etc.; (2) tree properties such that the number of subtrees of certain size, tree edit-distance, etc.; (3) many approximation algorithms on tree embeddings of an arbitrary metric~\cite{blelloch2017efficient,blelloch2012parallel}; and many more.} can be solved using dynamic programming, such that the computation essentially applies the 2-knapsack recurrence a hierarchical (bottom-up) manner.

We start by analyzing the lower bound on cache complexity of the 2-knapsack recurrence.
The computational grid has two dimensions, corresponding to $i$ and $j$ in the recurrence.
If we ignore $B$ in the recurrence, then the recurrence is identical to LWS (with no data dependencies), so we can apply the lower bounds in Section~\ref{sec:lws} here.

Note that each update requires two input values $A_{j}$ and $B_{i-j}$, but they are not independent.
When computing a subtask that corresponding to $(i,j)\in[i_0,i_0+n_i]\times[j_0,j_0+n_j]$, the projection sizes on input and output arrays $A$, $B$ and $D$ are no more than $n_j$, $n_i+n_j$ and $n_i$.
This indicates that the computation of this recurrence is a variant of \kdcomp{2}, so we can use the same algorithm discussed in Section~\ref{sec:algorithm}.
\begin{corollary}\label{thm:knapsack}
2-knapsack recurrence can be computed using $O(n^2)$ work, optimal symmetric and asymmetric cache complexity of $\displaystyle \Theta\left(n^2\over B{M}\right)$ and $\displaystyle \Theta\left(\wcost^{1/2}n^2\over B{M}\right)$, and $O(\log^2 n)$ \depth{}.
\end{corollary}




\section{Matrix Multiplication and All-Pair Shortest Paths}\label{sec:numerical}

In this section we discuss matrix multiplication, Kleene's algorithm on all pair shortest-paths, and some linear algebra algorithms including Strassen algorithm, Gaussian elimination (LU decomposition), and triangular system solver.
The common theme in these algorithms is that their computation structures are very similar to that of matrix multiplication, which is a \kdcomp{3}.
Strassen algorithm is slightly different and introduced separately in Appendix~\ref{sec:strassen}.
Other algorithms are summarized in Section~\ref{sec:linear} and the details are given in Appendix~\ref{sec:apsp}--\ref{sec:trs}.

We show improved asymmetric cache complexity for all problems.
For Gaussian elimination and triangular system solver, we show linear-depth race-free algorithms in both symmetric and asymmetric settings which are based on the parallel algorithm discussed in Section~\ref{sec:para}.
There exist work-optimal and sublinear depth algorithm for APSP~\cite{Tiskin01}, but we are unaware of how to make it I/O-efficient.
Compared to previous linear-\depth{} algorithms~\cite{tang2015cache,dinh2016extending}, our new algorithms are race-free and in the nested-parallel model.

\subsection{Matrix Multiplication}\label{sec:mm}

The combinatorial matrix multiplication (definition in Section~\ref{sec:prelim}) is one of the simplest cases of the \kdcomp{3}.
Given a semiring $(\times,+)$, in matrix multiplication each cell corresponds to a ``$\times$'' operation of the two corresponding input values and the ``$+$'' operation is associative.
Since there are no dependencies between the operations, we can simply apply Theorem~\ref{thm:sym-cost} and~\ref{thm:asym-cost} to get the following result.
\begin{corollary}\label{thm:mm}
Combinatorial matrix multiplication of size $n$ can be solved in $\Theta(n^3)$ work, optimal symmetric and asymmetric cache complexity of $\displaystyle \Theta\left(n^3\over B\sqrt{M}\right)$ and $\displaystyle \Theta\left(\wcost^{1/3}n^3\over B\sqrt{M}\right)$ respectively, and $O(\log^2 n)$ \depth{}.
\end{corollary}


The result for the symmetric case is well-known, but that for the asymmetric case is new.

\subsection{All-Pair Shortest Paths, Gaussian Elimination, and Triangular System Solver}\label{sec:linear}

We now discuss the well-known cache-oblivious algorithms to solve all-pair shortest paths (APSP) on a graph, Gaussian elimination (LU decomposition), and triangular system solver.
These algorithms share similar computation structures and can usually be discussed together.
Chowdhury and Ramachandran~\cite{chowdhury2006cache,chowdhury2010cache} categorized matrix multiplication, APSP, and Gaussian Elimination into the Gaussian Elimination Paradigm (GEP) and discussed a unified framework to analyze complexity, parallelism and actual performance.
We show how the parallel depth and the asymmetric cache complexity can be improved using the algorithms we just introduced in Section~\ref{sec:algorithm} and~\ref{sec:para}.

We discuss the details of these cache-oblivious algorithms in the appendix.
The common theme in these algorithms is that, the computation takes one or two square matrix(ces) of size $n\times n$ as input, applies $n^3$ operations, and generates output as a square matrix of size $n\times n$.
Each output entry is computed by an inner product of one column and one row of either the input matrices or the output matrix in some intermediate state.
Namely, the output $A_{i,j}$ requires input $B_{i,k}$ and $C_{k,j}$ for $1\le k\le n$ ($A$, $B$ and $C$ may or may not be the same matrix).
Therefore, we can apply the results of \kdcomp{3}s on these problems.\footnote{For Gaussian Elimination $A_{k,k}$ is also required, but $A_{k,k}$ is only on the diagonal, which requires a lower-order of cache complexity to load when computing a sub-cubic of a  3d grid.}
Note that some of the grids are full (e.g., Kleene's algorithm) while others are not, but they are all $\alpha$-full and  contain $O(n^3)$ operations.

The data dependencies in these algorithms are quite different from each other, but the recursions for cache complexity $Q(n)$ and depth $D(n)$ for APSP, Gaussian elimination and triangular system solver are all in the following form:
\[
\begin{split}
Q(n)&=\beta\, Q(n/2)+\gamma\, Q_{\smb{3C}}(n/2)\\
D(n)&=2\, D(n/2)+\delta\, D_{\smb{3C}}(n/2)
\end{split}\]
where $Q_{\smb{3C}}(n)$ and $D_{\smb{3C}}(n)$ are the cache complexity and depth of a \kdcomp{3} of size $n$.
Here, as long as the the recursive subtask fits into the cache together with the \kdcomp{3} computation and the $\beta$, $\gamma$ and $\delta$ are constants and satisfy $\beta<8$, we can show the following bounds.

\begin{theorem}
Kleene's algorithm for APSP, Gaussian elimination and triangular system solver of size $n$ can be computed in $\Theta(n^3)$ work, symmetric and asymmetric cache complexity of $\displaystyle O\left(\frac{n^3}{B\sqrt{M}}\right)$ and $\displaystyle O\left(\frac{\wcost{}^{1/3} n^3}{B\sqrt{M}}\right)$ respectively, and $O(n)$ depth.
\end{theorem}

We now discuss the cache-oblivious algorithms to solve all-pair shortest paths (APSP) on a graph with improved asymmetric cache complexity and linear span.
Regarding the \depth{}, Chowdhury and Ramachandran~\cite{chowdhury2010cache} showed an algorithm using $O(n\log^2 n)$ \depth{}.  
There exist work-optimal and sublinear \depth{} algorithm for APSP~\cite{Tiskin01}, but we are unaware of how to make it I/O-efficient while maintaining the same \depth{}.
Compared to previous linear \depth{} algorithms in~\cite{dinh2016extending}, our algorithm is race-free and in the classic nested-parallel model.
Also, we believe our algorithms are simpler.
The improvement is from plugging in the algorithms introduced in Section~\ref{sec:algorithm} and~\ref{sec:para} to Kleene's Algorithm.

An all-pair shortest-paths (APSP) problem takes a (usually directed) graph $G=(V,E)$ (with no negative cycles) as input.
Here we discuss the Kleene's algorithm (first mentioned in~\cite{kleene1951representation,munro1971efficient,fischer1971boolean,furman1970application}, discussed in full details in~\cite{Aho74}).
Kleene's algorithm has the same computational DAG as Floyd-Washall algorithm~\cite{Floyd:1962,warshall1962theorem}, but it is described in a divide-and-conquer approach, which is already I/O-efficient, cache-oblivious and highly parallelized.

\begin{algorithm}[!t]
\caption{\mf{Kleene}($A$)}
\label{alg:kleene}
\DontPrintSemicolon

\KwIn{Distance matrix $A$ initialized based on the input graph $G=(V,E)$}
\KwOut{Computed Distance matrix $A$}

\medskip

$A_{00}\gets \mf{Kleene}(A_{00})$\\
$A_{01}\gets A_{01}+A_{00}A_{01}$\\
$A_{10}\gets A_{10}+A_{10}A_{00}$\\
$A_{11}\gets A_{11}+A_{10}A_{01}$\\

\medskip

$A_{11}\gets \mf{Kleene}(A_{11})$\\
$A_{01}\gets A_{01}+A_{01}A_{11}$\\
$A_{10}\gets A_{10}+A_{11}A_{10}$\\
$A_{00}\gets A_{00}+A_{10}A_{01}$\\

\medskip

\Return{$A$}
\end{algorithm} 

The pseudocode of Kleene's algorithm is provided in Algorithm~\ref{alg:kleene}.
The matrix $A$ is partitioned into 4 submatrices indexed as $\begin{bmatrix}A_{00}&A_{01}\\A_{10}&A_{11}\end{bmatrix}$.
The matrix multiplication is defined in a closed semi-ring with $(+, \min)$.
Kleene's algorithm is a divide-and-conquer algorithm to compute APSP.
Its high-level idea is to first compute the APSP between the first half of the vertices only using the paths between these vertices.
Then by applying some matrix multiplication we update the shortest-paths between the second half of the vertices using the computed distances from the first half.
We then apply another recursive subtask on the second half vertices.
The computed distances are finalized, and we use them to again update the shortest-paths from the first-half vertices.

The cache complexity $Q(n)$ and \depth{} $D(n)$ of this algorithm follow the recursions:
\[
\begin{split}
Q(n)&=2Q(n/2)+6Q_{\smb{MM}}(n/2)\\
D(n)&=2D(n/2)+2D_{\smb{MM}}(n/2)
\end{split}\]
where $Q_{\smb{MM}}(n)$ is the I/O cost of a matrix multiplication of input size $n$.
The recursion of $Q(n)$ is root-dominated, which indicates that computing all-pair shortest paths of a graph has the same upper bound on cache complexity as matrix multiplication.

\begin{theorem}
Kleene's Algorithm to compute all-pair shortest paths of a graph of size $n$ uses $\Theta(n^3)$ work, has symmetric and asymmetric cache complexity of $\displaystyle \Theta\left(n^k\over M^{1/(k-1)}B\right)$ and $\displaystyle \Theta\left(n^k\wcost^{1/k}\over M^{1/(k-1)}B\right)$, and $O(n)$ \depth{}.
\end{theorem}

Similar to some other problems in this paper, the symmetric cache complexity is well-known, but the results in the asymmetric setting as well as the parallel approach are new.


\section{Conclusions and Future Work}\label{sec:conclusion}

In this paper, we shown improved cache-oblivious algorithm of many DP recurrences and in linear algebra, in the symmetric and asymmetric settings, both sequentially and in parallel. 
Our key approach is to show the correspondence between the recurrences and algorithms and the \ourcomp{}, and new results for computing the \ourcomp{}.
We believe that this abstraction provides a simpler and intuitive framework on better understanding these algorithms, proving lower bounds, and designing algorithms that are both I/O-efficient and highly parallelized.
It also provides a unified framework to bound the asymmetric cache complexity of these algorithms.

Based on the new perspective, we provide many new results, but we also observe many new open problems.
Among them are:
\begin{enumerate}
  \item The only non-optimal algorithm regarding cache complexity in this paper is for the GAP recurrence.  The I/O cost has an additional low-order term of $O((n^2\log M)/B)$.  Although in practice this term will not dominate the running time (the computation has $O(n^3)$ arithmetic operations), it is theoretically interesting to know if we can remove this term (even without the constraints of being cache-oblivious or based on divide-and-conquer).
  \item We show our algorithms in the asymmetric setting are optimal under the assumption of constant-branching (the CBCO paradigm).  Since the cache-oblivious algorithms discussed in this paper are leaf-dominate, we believe this assumption is always true.  We wonder if this assumption is necessary (i.e., if there exists a proof without using it, or if there are cache-oblivious algorithms on these problems with non-constant branching but still I/O-optimal).
  \hide{
  \item In Section~\ref{sec:tradeoff} we discuss how to trade higher I/O cost for lower depth.  It is interesting that all the small subtasks of different recurrences can be solved using some kind of the techniques based matrix multiplication, but checking or even trying to understand the details of these techniques is very complicated.  We wonder if we can solve the base cases in a unified framework.
  }
  \item The parallel symmetric cache complexity $Q_p$ on $p$ processors is $Q_1 + O(pDM/B)$, which is a loose upper bound when $D$ is large.  Although it might be hard to improve this bound on any general computation under randomized work-stealing, it can be a good direction to show tighter bounds on more regular computation structures like the \ourcomp{s} or other divide-and-conquer algorithms.  We conjecture that the additive term can be shown to be optimal (i.e., $O(pM/B)$) for the \ourcompfull{s}.
  \item Due to the page limit, in this paper we mainly discussed the lower bounds and algorithms for \sqrfull{s}, which is the setting of the problems in this paper (e.g., APSP, dynamic programming recurrences).  It is interesting to see a more general analysis on \ourcomp{s} with arbitrary shape, and such results may apply to other applications like the computation of tensor algebra.
\end{enumerate}

\section*{Acknowledgments}

This work is supported by the National
Science Foundation under CCF-1314590 and
CCF-1533858. Any opinions, findings, and conclusions
or recommendations expressed in this material are those of
the authors and do not necessarily reflect the views of the
National Science Foundation.
The authors thanks Yihan Sun for valuable discussions on various ideas, and Yihan Sun and Yuan Tang for the preliminary version of the algorithm in Section~\ref{sec:algorithm}. 

\bibliographystyle{plain}


\appendix

\section{Strassen Algorithm}\label{sec:strassen}

Strassen algorithm computes matrix multiplication on a ring.
Given two input matrices $A$ and $B$ and the output matrix $C=AB$, the algorithm partitions $A$, $B$ and $C$ into quadrants, applies seven recursive matrix multiplications on the sums or the differences of the quadrants, and each quadrant of $C$ can be calculated by summing a subset of the seven intermediate matrices.
This can be done in $O(n^{\log_27})$ work, $O(n^{\log_27}/M^{\log_47-1}B)$ cache complexity and $O(\log^2 n)$ depth.

Technically the computation structure of Strassen is not a \ourcomp{}, but we can apply a similar idea in Section~\ref{sec:algorithm} to reduce asymmetric cache complexity.
We still use $r$ as the balancing factor between reads and writes (set to be $\wcost{}^{2/3}$ in classic matrix multiplication).
Given square input matrices, the algorithm also partition the output into $r$-by-$r$ submatrices, and then run the 8-way divide-and-conquer approach to compute the matrix multiplication.
This gives the following recurrences on work ($T$), reads, writes and depth based on the output size $n$:
\[
\begin{split}
T'(n)&=7T'(n/2)+O(n^2)\\
R'(n)&=7R'(n/2)+O(n^2/B)\\
W'(n)&=7W'(n/2)+O(\wcost{}n^2/B)\\
D'(n)&=D(n/2)+O(\log n)
\end{split}\]
with the base cases $T'(1)=r$, $R'(\sqrt{rM})=M/B$, $W'(\sqrt{M})=\wcost{}M/B$, and $D'(1)=1$.
They solve to
\[
\begin{split}
R(n)=r^2R'(n/r)= &\ O(n^{\log_27}r^{2-\log_47}M^{1-\log_47}/B)\\
W(n)=r^2W'(n/r)= &\ O(\wcost{}n^{\log_27}r^{2-\log_27}M^{1-\log_47}/B)\\
D(n)= &\ O(\log^2n)
\end{split}\]
The case when $r=\wcost^{\log_74}$ gives the minimized cache complexity of $$Q(n)=O(n^{\log_27}\wcost^{\log_716-1}M^{1-\log_47}/B)\approx O(n^{2.8}\wcost^{0.42}/BM^{0.4})$$ an $O(\wcost^{0.58})$ improvement over the non-write-efficient version.
In this setting the work is $O(n^{\log_27}\wcost{}^{\log_764-2})$, a factor of $O(\wcost^{0.14})$ or $O(\wcost^{1/7})$ extra work. 

\section{All-Pair Shortest-Paths (APSP)} \label{sec:apsp}

We now discuss the cache-oblivious algorithms to solve all-pair shortest paths (APSP) on a graph with improved asymmetric cache complexity and linear span.
Regarding the \depth{}, Chowdhury and Ramachandran~\cite{chowdhury2010cache} showed an algorithm using $O(n\log^2 n)$ \depth{}.  
There exist work-optimal and sublinear \depth{} algorithm for APSP~\cite{Tiskin01}, but we are unaware of how to make it I/O-efficient while maintaining the same \depth{}.
Compared to previous linear \depth{} algorithms in~\cite{dinh2016extending}, our algorithm is race-free and in the classic nested-parallel model.
Also, we believe our algorithms are simpler.
The improvement is from plugging in the algorithms introduced in Section~\ref{sec:algorithm} and~\ref{sec:para} to Kleene's Algorithm.

An all-pair shortest-paths (APSP) problem takes a (usually directed) graph $G=(V,E)$ (with no negative cycles) as input.
Here we discuss the Kleene's algorithm (first mentioned in~\cite{kleene1951representation,munro1971efficient,fischer1971boolean,furman1970application}, discussed in full details in~\cite{Aho74}).
Kleene's algorithm has the same computational DAG as Floyd-Washall algorithm~\cite{Floyd:1962,warshall1962theorem}, but it is described in a divide-and-conquer approach, which is already I/O-efficient, cache-oblivious and highly parallelized.

\begin{algorithm}[!t]
\caption{\mf{Kleene}($A$)}
\label{alg:kleene}
\DontPrintSemicolon

\KwIn{Distance matrix $A$ initialized based on the input graph $G=(V,E)$}
\KwOut{Computed Distance matrix $A$}

\medskip

$A_{00}\gets \mf{Kleene}(A_{00})$\\
$A_{01}\gets A_{01}+A_{00}A_{01}$\\
$A_{10}\gets A_{10}+A_{10}A_{00}$\\
$A_{11}\gets A_{11}+A_{10}A_{01}$\\

\medskip

$A_{11}\gets \mf{Kleene}(A_{11})$\\
$A_{01}\gets A_{01}+A_{01}A_{11}$\\
$A_{10}\gets A_{10}+A_{11}A_{10}$\\
$A_{00}\gets A_{00}+A_{10}A_{01}$\\

\medskip

\Return{$A$}
\end{algorithm} 

The pseudocode of Kleene's algorithm is provided in Algorithm~\ref{alg:kleene}.
The matrix $A$ is partitioned into 4 submatrices indexed as $\begin{bmatrix}A_{00}&A_{01}\\A_{10}&A_{11}\end{bmatrix}$.
The matrix multiplication is defined in a closed semi-ring with $(+, \min)$.
Kleene's algorithm is a divide-and-conquer algorithm to compute APSP.
Its high-level idea is to first compute the APSP between the first half of the vertices only using the paths between these vertices.
Then by applying some matrix multiplication we update the shortest-paths between the second half of the vertices using the computed distances from the first half.
We then apply another recursive subtask on the second half vertices.
The computed distances are finalized, and we use them to again update the shortest-paths from the first-half vertices.

The asymmetric cache complexity $Q(n)$ of this algorithm follows the recursion of:
\[
\begin{split}
Q(n)&=2Q(n/2)+6Q_{\smb{3C}}(n/2)\\
D(n)&=2D(n/2)+2D_{\smb{3C}}(n/2)
\end{split}\]
Considering this cost, the recursion is root-dominated, which indicates that computing all-pair shortest paths of a graph has the same upper bound on cache complexity as matrix multiplication.

\hide{
\begin{corollary}
All-pair shortest paths of a graph of size $n$ can be computed in $\displaystyle \Theta\left(n^k\over M^{1/(k-1)}B\right)$ and $\displaystyle \Theta\left(n^k\wcost^{1/k}\over M^{1/(k-1)}B\right)$ work respectively for the symmetric and \anp{} model, and $O(n)$ depth.
\end{corollary}
}

\section{Gaussian Elimination}\label{sec:ge}

Gaussian elimination (without pivoting) is used in solving of systems of linear equations and computing LU
decomposition of symmetric positive-definite or diagonally dominant real matrices.
Given a linear system $AX=b$, the algorithm proceeds in two phases.
The first phase modifies $A$ into an upper triangular matrix (updates $B$ accordingly), which is discussed in this section.
The second phase solves the values of the variables using back substitution, which is shown in Section~\ref{sec:trs}.

The process of Gaussian elimination can be viewed as a three nested-loops and computing the value of $A_{i,j}$ requires $A_{i,k}$, $A_{k,j}$ and $A_{k,k}$ for all $1\le k<i$.
If required, the corresponding value of the LU decomposition matrix can be computed simultaneously.
The underlying idea of divide-and-conquer approach is almost identical to Kleene's algorithm,
which partitions $A$ into four quadrants $\begin{bmatrix}A_{00}&A_{01}\\A_{10}&A_{11}\end{bmatrix}$.
The algorithm: (1) recursively computes $A_{00}$; (2) updates $A_{10}$ and $A_{01}$ using $A_{00}$; (3) updates $A_{11}$ using $A_{10}$ and $A_{01}$; and (4) recursively computes $A_{11}$.

Note that each inter-quadrant update in step (2) and (3) is a \kdcomp{3}, which gives the following recurrence:
\begin{align*}
Q(n)&=2Q(n/2)+4Q_{\smb{3C}}(n/2)\\
D(n)&=2D(n/2)+3D_{\smb{3C}}(n/2)
\end{align*}

\section{Triangular System Solver}\label{sec:trs}

A Triangular System Solver computes the back substitution step in solving the linear system.
Here we assume that it takes as input a lower triangular $n\times n$
matrix $T$ (can be computed using the algorithm discussed in Section~\ref{sec:ge}) and a square matrix $B$ and outputs a square matrix $X$ such
that $TX = B$. A triangular system can be recursively decomposed
as:
\[
\begin{split}
\begin{bmatrix}B_{00} & B_{01} \\ B_{10} & B_{11} \end{bmatrix}&=
\begin{bmatrix}T_{00} & 0 \\ T_{10} & T_{11} \end{bmatrix}
\begin{bmatrix}X_{00} & X_{01} \\ X_{10} & X_{11} \end{bmatrix}\\&=
\begin{bmatrix}T_{00}X_{00} & T_{00}X_{01} \\ T_{10}X_{00}+T_{11}X_{10} & T_{10}X_{01}+T_{11}X_{11} \end{bmatrix}
\end{split}\]
such that four equally sized subquadrants $X_{00}$, $X_{01}$, $X_{10}$, and $X_{11}$ can be solved recursively.
In terms of parallelism, the two subtasks of $X_{00}$ and $X_{01}$ are independent, and need to be solved prior to the other independent subtasks $X_{10}$, and $X_{11}$.

The asymmetric cache complexity $Q(n)$ of this algorithm follows the recursion of:
\begin{align*}
Q(n)&=4Q(n/2)+2Q_{\smb{3C}}(n/2)\\
D(n)&=2D(n/2)+D_{\smb{3C}}(n/2)
\end{align*}
\hide{
where $Q_{\smb{MM}}(n)$ is the I/O cost of a matrix multiplication of input size $n$.
Considering this cost, the recursion is root-dominated, which indicates that a triangular system solver has the same upper bound on I/O cost as matrix multiplication.
}

\hide{
\subsection{Cholesky Decomposition}

Cholesky Decomposition takes the input as an $n\times n$ Hermitian,
positive-definite matrix $A$, and computes another $n\times n$ lower triangular matrix $L$ such that $A = LL^T$.
The following 2-way divide-and-conquer recurrence solves the problem:
\[
\begin{bmatrix}A_{00} & A_{10}^T \\ A_{10} & A_{11} \end{bmatrix}=
\begin{bmatrix}L_{00}L^T_{00} & L_{00}L^T_{10} \\ L_{10}L^T_{00} & L_{10}L^T_{01}+L_{11}L^T_{11} \end{bmatrix}
\]
To compute this, the top-left quadrant can be solved recursively, then the top-right quadrant can be computed using a Triangular System Solver, and lastly the bottom-right quadrant can be solved recursively again after a matrix multiplication and another substraction.

The asymmetric cache complexity $Q(n)$ of this algorithm follows the recursion of:
\begin{align*}
Q(n)=2Q(n/2)+Q_{\smb{TRS}}(n/2)+Q_{\smb{3C}}(n/2)\\
D(n)=2D(n/2)+D_{\smb{TRS}}(n/2)+D_{\smb{3C}}(n/2)
\end{align*}
}
\hide{
where $Q_{\smb{TRS}}(n)$ and $Q_{\smb{MM}}(n)$ are the I/O costs of a Triangular System Solver and matrix multiplication of input size $n$.
The recursion is root-dominated, which indicates that Cholesky Decomposition has the same upper bound on I/O cost as matrix multiplication.
}

\section{Discussions of Carson et al.~\cite{carson2016write}}\label{sec:carson}

Carson et al.\ also analyze algorithms that reduce the number of writes~\cite{carson2016write}.
They concluded that many cache-oblivious algorithms like matrix multiplication could not be write-avoiding.
Their definition of write-avoiding is different from our write-efficiency, and it requires the algorithm to reduce writes without asymptotically increasing reads.
Hence, their negative conclusion does not contradict the result in this paper.   

We now use matrix multiplication as an example that optimal number of reads leads to worse overall asymmetric cache complexity.
Let's say a write is $n$ times more expansive than a read.
One algorithm can apply $n^2$ inner products, and the asymmetric cache complexity is $O(n^3/B)$: $O(n/B)$ reads and $O(1/B)$ writes per inner product.
However, the algorithms using optimal number of reads requires $O(n^3/B\sqrt{M})$ reads and writes, so the overall cache complexity is $O(n\cdot n^3/B\sqrt{M})$.
The first algorithm requiring more reads is a factor of $O(n/\sqrt{M})$ better on the asymmetric cache complexity.
Notice that we always assume $n=\omega(\sqrt{M})$ since otherwise the whole computation is trivially in the cache and has no cost.
As a result, an algorithm with a good asymmetric cache complexity does not always need to be write-avoiding.

The algorithms on asymmetric memory in this paper all require extra reads, but can greatly reduce the overall asymmetric cache complexity compared to the previous cache-oblivious algorithms.
The goal of this paper is to find the optimal cache-oblivious algorithms for any given write-read asymmetry $\wcost{}$.

\end{document}